\documentclass[11pt]{article}

\usepackage{amssymb} 
\usepackage{amsmath}     
\usepackage{amsthm}
 \usepackage{graphicx}  
 \usepackage{a4wide} 
 
\usepackage[ruled,vlined]{algorithm2e}
\SetKw{KwStep}{Step}

\newtheorem{prop}{Proposition}
\newtheorem{thm}{Theorem}

\newcommand{\Rset}{\mathbb{R}}
\newcommand{\Xset}{\mathbb{X}}

\newenvironment{keyword}{\par{\noindent\bf Keywords:}}

\begin{document}

\title{A  Robust Lot Sizing Problem with Ill-known Demands
\thanks{The second author of the paper 
was partially supported by 
Polish Committee for Scientific Research, grant
 N~N206~492938.}
}

\author{
Romain    Guillaume,\\
{\small \textit{Universit{\'e} de Toulouse-IRIT,}}\\
{\small \textit{5, All{\'e}es A. Machado,}}
{\small \textit{31058 Toulouse Cedex 1, France,}}\\
{\small \texttt{Romain.Guillaume@irit.fr}}
  \and
  Przemys{\l}aw Kobyla{\'n}ski, Pawe{\l} Zieli{\'n}ski\footnote{Corresponding author} \\
{\small \textit{Faculty of Fundamental Problems of Technology,}}\\
{\small \textit{Wroc{\l}aw University of Technology,}}\\
{\small \textit{Wybrze{\.z}e Wyspia{\'n}skiego 27,}}\\
{\small \textit{50-370 Wroc{\l}aw, Poland,}}\\
{\small \texttt{\{Przemyslaw.Kobylanski,Pawel.Zielinski\}@pwr.wroc.pl}}
}

\date{}

\maketitle

\begin{abstract}
The paper deals with a  lot sizing
 problem with ill-known demands
 modeled by fuzzy intervals whose membership functions are
 possibility distributions for the values of the uncertain demands.
  Optimization criteria, in the setting of possibility theory, that
  lead to choose robust production plans under fuzzy demands are given. 
  Some algorithms for determining optimal robust production plans
   with respect to the proposed criteria,  and for evaluating production plans are provided.
  Some computational experiments are presented.
\end{abstract}

\begin{keyword}
 Fuzzy Optimization,
Dynamic Lot Sizing,
 Uncertain Demand,
  Possibility Theory
\end{keyword}

\section{Introduction}

Nowadays, companies do not  compete as independent
entities but as a part of collaborative supply chains. 
Uncertainty in  demands creates a risk in a supply chain as
backordering, obsolete inventory due to the bullwhip effect~\cite{LPW97}.
To reduce this risk two different approaches exist that are considered here.
The first approach consists in
a collaboration between the customer and the supplier and the second
one consists in an integration
of uncertainty into a planning process.

The collaborative processes  mainly aim to reduce a risk in  a supply chain~\cite{GC09}.
This is done by
enforcing
a coordination in  a supply chain.
Two approaches can be applied: vertical and horizontal.
The vertical approach is a centralized decision making that synchronizes a supply chain 
(the most common
way to coordinate within companies). 
The horizontal one refers  to the collaborative
planning,   in which
a supply chain can be seen as a chain,
where actors are
independent entities~\cite{D09}.
The industrial collaborative planning 
 has been standardized for implementing a cooperation
between retailers and manufactures. This process is called \emph{Collaborative Planning,
Forecasting and Replenishment} (CPFR)~\cite{IC05}. More precisely, the
collaborative processes are usually characterized by a set of point-to-point (customer/supplier)
relationships with a partial information sharing~\cite{GC09,ASRS09}. 
In the  collaborative supply chain, a
procurement plan is built and propagated through a supply chain. Namely, 
the procurement plan is composed of three horizons:  freezing, flexible and free ones~\cite{GC09}. Quantities in the freezing horizon are crisp and can not be modified, quantities
 in the flexible horizon are intervals and can be modified under constraints 
imposed by a previous
procurement plan.  In the  free horizon quantities can be modified without constraints.
Another way to reduce a risk in a supply chain is to integrate
the uncertainty in a planning process. 
In the literature, 
three different sources of uncertainty are
distinguished  (see~\cite{PMPL09} for a review): \emph{demand}, \emph{process} 
and \emph{supply}.  
These uncertainties  are  due to  difficulties to access to available historical data allowing to 
determine a probability distribution.

In this paper, we focus on the collaborative supply chain (a supply chain,
 where actors are independents entities) 
 under uncertain demands. 
  In most companies today, especially in
aeronautic companies,  actors use the 
\emph{Manufacturing Resource Planning} (MRPII) to plan theirs production. 
MRPII is a
planning control process composed of three processes (the production process, 
the procurement process and the distribution process)
and three levels~\cite{ACC11}: the strategic level (computing   commercial
and industrial plans), the tactical level (the \emph{Master Production Scheduling}
 (MPS) and the \emph{Material
Requirement Planning} (MRP)) and the operational level (a detailed scheduling and 
a shop floor control). 
MRPII have been also  extended  to take into account:
the imprecision on  quantities   of
demands (MPS)~\cite{HT00}, 
the imprecision on quantities  of   demands and uncertain orders~\cite{GLRV05} (MRP)
and the imprecision on 
 quantities and on dates of demands with uncertain order dates~\cite{GTG10} (MRP).

 In this paper, we wish to investigate the part of
the MRPII process. Namely, 
 the procurement process in the tactical level in the collaborative context. 
 Our purpose is 
 to help the decision maker of a procurement service to
evaluate a performance of a given procurement plan with ill-known gross requirements
 and to
compute a procurement plan in a collaborative supply chain (with and without
supplier capacity sharing due to a procurement contract)  with ill-known gross requirements.

Several
production planning problems have been adapted  to the case of fuzzy demands:
\emph{economic order quantity}~\cite{KH02,PKMM08}, 
\emph{multi-period planning}~\cite{HT00,GLRV05,GTG10,SY08,YLS09,MPG07,MPP10,TRGZ07}, and the \emph{problem of supply chain planning}
(production distribution, centralized supply chain)~\cite{AFGA07,LC09,PMPV09,PMJB10,PRP99,WS05}.
 In the literature, 
 there are two popular
 families of approaches  for coping with fuzzy parameters.
 In the first family,  a \emph{defuzzification} is first performed 
 and then deterministic optimization methods are used~\cite{PMPV09,PMJB10}.
 In the second one,
the objective is expressed in the setting of 
\emph{possibility theory}~\cite{DP88}  and \emph{credibility theory}~\cite{L04}.
We can distinguish: the possibilistic programming  (a~fuzzy mathematical programming) in which
a solution optimizing
 a criterion
 based on the \emph{possibility} 
  measure is built~\cite{MPP10,TRGZ07},
  the \emph{credibility} measure based programming in which
   the credibility measure is used to guaranty a service level 
   (chance constraints on the inventory level)~\cite{LLS09} or
  the goal is to choose a solution that optimizes a criterion
 based on the credibility measure~\cite{SY08}
 and  a decision support based on the propagation of the uncertainty to the inventory level and backordering level~\cite{HT00,GLRV05,GTG10}.
 Here, we restrict our attention to uncertainty propagation 
 in MRP (the tactical level)~\cite{HT00,GLRV05,GTG10} 
 and we propose methods both for  evaluating a procurement
plan in terms of costs under uncertain demands and for computing 
a procurement plan
 which minimizes
the impact of uncertainty on costs,
since
the approaches proposed in the literature 
are not able to do  this.

Popular 
setting of problems for hedging against 
uncertainty  of parameters is \emph{robust optimization}~\cite{KY97}.
In the robust optimization setting the uncertainty is modeled by specifying a
set of all possible realizations of the parameters called \emph{scenarios}. No probability distribution in the scenario set is given. 
The value of each parameter may fall within a given closed
interval and the set of scenarios is the Cartesian product of
these intervals.
Then, in order to choose a solution,
two optimization criteria,
called the \emph{min-max} and  the \emph{min-max regret}, can be
adopted. 
Under the min-max criterion, we seek a solution that minimizes 
the largest cost over all scenarios. Under the min-max regret criterion we wish to find a solution, which minimizes the largest deviation from optimum over all scenarios.

In this paper, we are interested in  computing
a robust procurement plan (with and without delivering capacity
of the supplier sharing). The delivering capacity are composed of
two bounds: the lower one being  the minimal accepted quantity
that is sent to the customer and the upper bound which is due to a
production capacity of the supplier. Moreover the customer
accepts to have backordering but it is more penalized that
inventory.
This problem is equivalent to the \emph{problem of production
planning} with backordering, more precisely to a certain version of
the \emph{lot sizing problem} (see, e.g.,~\cite{WW58,FLK80}),   where: the procured quantities are
production quantities, a \emph{production plan}; delivering constraints are production
constraints, \emph{capacity limits on production plans}; and the gross requirements are  \emph{demands}.
Thus, the problem consists in
finding a  production plan  that 
fulfills capacity limits and 
 minimizes the total cost of storage and backordering subject to the conditions of satisfying each demand.
It is efficiently
solvable when the demands are precisely
known (see, e.g.,~\cite{DK97,KGW03,AH08}).
However, the demands are seldom precisely known in advance
and the uncertainty must be taken into account.

In this paper, we consider the above problem
with uncertain demands modeled by fuzzy intervals. The membership function
of a fuzzy interval  is a \emph{possibility distribution}
describing, for each value of the demand, the extent
to which it is a possible value. 
In other words, it means that
the value of this demand belongs to a $\lambda$-cut of the fuzzy interval 
with the degree of necessity (confidence)
$1-\lambda$. 
To evaluate  a production plan, we assign to it,
 \emph{degrees of possibility and
necessity} that  its cost does not exceed a given threshold and
a degree of necessity
that costs of the plan fall within 
a given fuzzy goal.
In order to find ``robust solutions'' under fuzzy demands,
we apply two criteria.
The first one consists in choosing
 a production plan
which maximizes 
the degree of necessity (certainty)   that its cost does not exceed  a given threshold.
The second criterion is
weaker than the first one and consists   
in choosing a plan with
 the maximum degree of necessity
that costs of the plan fall within 
a given fuzzy goal. A similar criterion has been proposed in~\cite{KK09} for
discrete optimization problems with fuzzy costs.
We provide some methods for finding a  robust production plan
 with respect to the proposed criteria as well as for evaluating a given production plan
under fuzzy-valued demands which heavy rely on methods for 
 finding a robust production plan, called 
\emph{optimal robust production plan}, in 
the problem of production
planning under interval-valued demands with the robust  min-max criterion.
Namely,
it turns out that  the considered fuzzy problems can be reduced to examining a family of the interval
problems with the min-max criterion. Therefore, we generalize in this way the  min-max criterion 
under the interval structure of uncertainty to the fuzzy   case.

The paper is organized as follows. In
Section~\ref{ssnpt}, we recall some 
notions of possibility theory.  In
Section~\ref{sdlsp}, we present 
a lot-size problem with backorders and precise demands. 
In Section~\ref{srob}, we present our results. Namely, we
 investigate the interval case, that is the lot-size problem with backorders
in which uncertain demands are specified as closed intervals. We construct
 algorithms  for 
 finding an optimal robust production plan (a polynomial algorithm for the case without
 capacity limits and an iterative algorithm for the case with capacity limits)
 and for
 evaluating a given production plan (linear and mixed integer programing methods,
 a pseudopolynomial algorithm).
 An experimental evidence of the efficiency of the proposed algorithms is provided.
 In Section~\ref{sfp}, we extend our results from the previous section to the fuzzy case.
 We study the lot-size problem with backorders with uncertain
 demands modeled by fuzzy intervals in a setting of possibility theory.
 We provide methods for seeking a  robust production plan
 with respect to two proposed criteria as well as for evaluating a given production plan
under fuzzy-valued demands (the methods heavily rely on the ones from the interval case).
 The efficiency of the methods is confirmed experimentally. 

\section{Selected Notions of Possibility Theory}
\label{ssnpt}

A \emph{fuzzy interval} $\widetilde{A}$ is a fuzzy set in $\Rset$ whose membership function $\mu_{\widetilde{A}}$ is normal, quasi concave and upper semicontinuous. Usually, it is  assumed that the support of a fuzzy interval  is bounded. The main property of a fuzzy interval is the fact that all its $\lambda$-cuts, that is the sets $\widetilde{A}^{[\lambda]}=\{x:\mu_{\widetilde{A}}(x)\geq \lambda\}$, $\lambda \in (0,1]$, are closed intervals. We will assume that $\widetilde{A}^{[0]}$ is the smallest closed set containing the support of $\widetilde{A}$. So, every fuzzy interval $\widetilde{A}$ can be represented as a family of closed intervals $\widetilde{A}^{[\lambda]}=[a^{-[\lambda]},a^{+[\lambda]}]$, parametrized by the value of $\lambda\in [0,1]$. In many applications, the class of \emph{trapezoidal fuzzy intervals} is used.
 A trapezoidal fuzzy interval, denoted by a quadruple $\widetilde{A}=(a,b,c,d)$
 and  its membership function has the following form:
 \[
\mu_{\widetilde{A}}(z)=
\begin{cases}
	0 &\text{if $z\leq a$,}\\
	 \frac{z-a}{b-a} &\text{if $a< z < b $,}\\
	 1 &\text{if $b\leq z \leq c $,}\\
	 \frac{d-z}{d-c} &\text{if $c< z < d $,}\\
	 0 &\text{if $z\geq d$.}
\end{cases}
\]
Its  $\lambda$-cuts are simply
$[a+\lambda(b-a), d-\lambda (d-c)]$ for $\lambda\in [0,1]$. Notice that this representation contains  \emph{triangular fuzzy intervals}  ($b=c$).
 
Let us now recall the possibilistic interpretation of fuzzy intervals.
\emph{Possibility theory}~\cite{DP88} is an approach to handle  incomplete
information and it
relies on two dual measures: \emph{possibility} and
\emph{necessity}, which express plausibility and certainty of events.
 Both measures
are built from a \emph{possibility distribution}. 
Let a fuzzy interval $\widetilde{A}$ be attached with 
a single-valued variable~$a$ (uncertain real quantity). 
The membership function $\mu_{\widetilde{A}}$ is understood as
a possibility distribution, $\pi_{a}=\mu_{\widetilde{A}}$, which 
describes the set of more or less plausible,
mutually exclusive values of the variable~$a$. 
It  can encode a family of probability functions~\cite{DP92}. In particular, a degree of possibility can be viewed as
the upper bound of a degree of probability~\cite{DP92}.
The value of $\pi_{a}(v)$  represents the possibility degree of the assignment~$a=v$, i.e.
$\mathrm{\Pi}(a=v)=
\pi_{a}(v)=\mu_{\widetilde{A}}(v)$,
 where $\mathrm{\Pi}(a=v)$ is the possibility of the event that 
 $a$ will take the value of $v$.
In particular, $\pi_{a}(v)=0$ 
means that $a=v$ is impossible 
and $\pi_{a}(v)>0$ means that $a=v$ is plausible. 
Equivalently,
 it means that the value of $a$ 
belongs to a $\lambda$-cut $\widetilde{A}^{[\lambda]}$ with
confidence (or degree of necessity) $1-\lambda$.
A detailed interpretation of the possibility distribution and some methods of obtaining it from the possessed knowledge are described in~\cite{DP88,DP94}. 
Let $\widetilde{G}$ be a fuzzy interval. Then ``$a\in \widetilde{G}$''
is a \emph{fuzzy event}.
The \emph{possibility} of ``$a \in \widetilde{G}$'', 
denoted by $\Pi(a\in \widetilde{G})$, is
as follows~\cite{DFF03}: 
\begin{equation}
\Pi(a \in \widetilde{G})=\sup_{v\in \Rset}\min\{\pi_{a}(v),\mu_{\widetilde{G}}(v)\}.
\label{dposs}
\end{equation}
$\Pi(a \in \widetilde{G})$ evaluates the extent to which  
``$a \in \widetilde{G}$''  is possibly true. 
The \emph{necessity}   of  event
``$a \in \widetilde{G}$'', denoted by 
$\mathrm{N}(a\in \widetilde{G})$, is
as follows: 
\begin{align}
\mathrm{N}(a \in \widetilde{G})&=1-\Pi(a\not\in \widetilde{G})
=1-\sup_{v\in \Rset}\min\{\pi_{a}(v),1-\mu_{\widetilde{G}}(v)\}\label{dnec}\\
&=\inf_{v\in\Rset}\max\{1-\pi_{a}(v),\mu_{\widetilde{G}}(v)\},\nonumber
\end{align}
where $1-\mu_{\widetilde{G}}$ is the membership function
of the complement of the fuzzy set~$\widetilde{G}$. 
$\mathrm{N}(a \in \widetilde{G})$ evaluates the extent to which  
``$a \in \widetilde{G}$'' is certainly true. Observe that if $G$ is a classical set, then $\Pi(a\in G)=\sup_{v\in G} \pi_{a}(v)$ and $\mathrm{N}(a\in G)=1-\sup_{v\notin G}\pi_{a}(v)$.

\section{The Deterministic Problem}
\label{sdlsp}

We are given $T$ periods. For period~$t$, $t=1,\ldots,T$,
let $d_t$ be the demand in period~$t$, $d_t\geq 0$
(here we assume that the demands are precise), $x_t$ the production
amount in period~$t$, $x_t\geq 0$, $l_t$, $u_t$ the production capacity limits on~$x_t$. 
Let $\Xset\subseteq \Rset^{T}_{+}$
be the set of feasible production amounts.
Two cases are distinguished, the case
\emph{with no capacity limits} $\Xset=\{(x_1,\ldots,x_{T})\;:\;
x_t\geq 0, t=1,\ldots,T\}$
and the one \emph{with capacity limits}
$\Xset=\{(x_1,\ldots,x_{T})\;:\;
l_t\leq x_t\leq u_t, t=1,\ldots,T\}$.
Set $\mathbf{D}_t=\sum_{i=1}^{t}d_i$ and
 $\mathbf{X}_t=\sum_{i=1}^{t}x_i$,
 $\mathbf{D}_t$ and  $\mathbf{X}_t$ stand for
 the cumulative demand up to period~$t$
and  the production level up to period~$t$, respectively.
Obviously, $\mathbf{X}_{t-1}\leq \mathbf{X}_t$ and 
$\mathbf{D}_{t-1}\leq \mathbf{D}_t$, $t=2,\ldots,T$.
The costs of carrying one unit of 
inventory from period~$t$ to period~$t+1$ is given by~$c^I_t\geq 0$
and
the costs of backordering one unit 
from period~$t+1$ to period~$t$ is given by $c^B_t\geq 0$. 
The nonnegative real function $L_t(u,v)$ represents either the cost of storing inventory 
from period~$t$ to period~$t+1$
or  the cost of backordering quantity from period~$t+1$ to
period~$t$,
namely
$L_t(\mathbf{X}_t,\mathbf{D}_t)=c^I_t(\mathbf{X}_t-\mathbf{D}_t)$
if $\mathbf{X}_t\geq \mathbf{D}_t$; $c^{B}_t(\mathbf{D}_t-\mathbf{X}_t)$ otherwise.
The function has the form 
$L_t(\mathbf{X}_t,\mathbf{D}_t)=\max\{c^{I}_t(\mathbf{X}_t-\mathbf{D}_t),
c^{B}_t(\mathbf{D}_t-\mathbf{X}_t)\}$.

Our production planning problem with the deterministic (precise) demands consists in
finding a feasible production plan  $\pmb{x}=(x_1,\ldots,x_{T})$,
$\pmb{x}\in \Xset$, that  minimizes the total cost of
storage and backordering subject to the conditions of satisfying each demand, namely
\begin{equation}
\min_{\pmb{x}\in \Xset}F(\pmb{x})=
\min_{\pmb{x}\in \Xset}\sum_{t=1}^{T}L_t(\mathbf{X}_t,\mathbf{D}_t).
\label{ddls}
\end{equation}
Obviously, the problem~(\ref{ddls}) is a version of the \emph{classical
dynamic lot-size problem with backorders} (see, e.g.,~\cite{WW58,FLK80}).
 Without loss of generality, we can assume that
an initial inventory~$I_0$ and an initial backorder~$B_0$
are equal to zero. Otherwise, one can append period~$0$
and assign $x_0=I_0$ and $d_0=0$ 
with zero inventory cost  
if $I_0>0$ or 
assign $x_0=0$ and $d_0=B_0$
with zero backorder cost
if $B_0>0$. 

In the case with no capacity limits problem~(\ref{ddls})
has a trivial optimal solution equal to $(d_1,\ldots,d_T)$.
In the case with capacity limits, (\ref{ddls}) can be formulated as
the minimum cost flow problem 
(see,~e.g.,~\cite{AMO93}):
\begin{equation}
 \begin{array}{lll}
 \min & \displaystyle\sum_{t=1}^{T}(c^I_t I_t+ c^B_t B_t)&\\
\text{s.t. } & B_t- I_t=\sum_{j=1}^{t}(d_j-x_j), & t=1,\ldots,T,\\
       &l_t\leq x_t\leq u_t,  &t=1,\ldots,T,\\
       &B_t,I_t\geq 0,  &t=1,\ldots,T.
 \end{array}
 \label{ddlsf}
\end{equation}
Problem~(\ref{ddlsf}) can  efficiently solved, for instance, by 
an algorithm presented in~\cite{AH08}
 that takes into account  a special structure of the underling network.

\section{Robust Problem}
\label{srob}
Assume that demands $d_t$, $t=1,\ldots,T$, in problem~(\ref{ddls}), are 
only known to belong to intervals $D_{t}=[d^{-}_{t},d^{+}_{t}]$,
$d^{-}_{t}\geq 0$.
This means that we neither know the exact 
demands, nor can we set them precisely. 
We assume that the demands are unrelated to one another
and there is no probability distribution in $D_t$, $t=1,\ldots,T$.
 A vector $S=(d_1,\ldots,d_T)$, $d_t\in D_t$, that represents an assignment of demands $d_t$ to periods~$t$, $t=1,\ldots,T$, is
called a \emph{scenario}. Thus every scenario
  expresses a realization of the demands.
We denote by $\Gamma$
   the set of all the scenarios, i.e. $\Gamma 
   =[d^{-}_{1},d^{+}_{1}]\times\cdots\times [d^{-}_{T},d^{+}_{T}]$.
   Among the scenarios of $\Gamma$
 \emph{extreme scenarios} can be distinguished,
 that is the ones, which belong to 
$\{d^{-}_{1},d^{+}_{1}\}\times\cdots\times \{d^{-}_{T},d^{+}_{T}\}$,
the set of extreme scenarios is denoted by $\Gamma_{\text{ext}}$.
We denote by $S^{+}$ (resp. $S^{-}$)
the extreme scenario in which all the demands
are set to their upper (resp. lower) bounds.
 The demand  and the cumulative demand in period~$t$ 
 under scenario~$S$ are denoted by 
 $d_t(S)\in D_t$ and $\mathbf{D}_t(S)$, respectively,
 $\mathbf{D}_t(S)=\sum_{i=1}^{t}d_i(S)$. Clearly, for every $S\in \Gamma$ it holds 
 $\mathbf{D}_{t-1}(S)\leq \mathbf{D}_t(S)$, $t=2,\ldots, T$, and
 $\mathbf{D}_t(S)\in [\mathbf{D}_t(S^{-}), \mathbf{D}_t(S^{+})]$.
 The function 
 $L_t(\mathbf{X}_t,\mathbf{D}_t(S))=
 \max\{c^{I}_t(\mathbf{X}_t-\mathbf{D}_t(S)),
c^{B}_t(\mathbf{D}_t(S)-\mathbf{X}_t)\}$,
 represents either the cost of storing inventory 
from period~$t$ to period~$t+1$
or  the cost of backordering quantity from period~$t+1$ to
period~$t$ under scenario~$S$.
Now $F(\pmb{x},S)$ denotes the total cost of 
a production plan $\pmb{x}\in \Xset$ under scenario~$S$, i.e.
$F(\pmb{x},S)=
\sum_{t=1}^{T}L_t(\mathbf{X}_t,\mathbf{D}_t(S))$.

In order to choose a robust production plan,
one of robust criteria,
called the \emph{min-max} can be
adopted (see, e.g.~\cite{KY97}).
In the \emph{min-max} version of problem~(\ref{ddls}), we seek a 
feasible production plan with the minimum
the worst total cost
over all scenarios, that~is
\[
\textsc{ROB}:\; \min_{\pmb{x}\in \Xset}A(\pmb{x})=
\min_{\pmb{x}\in \Xset}\max_{S\in\Gamma} F(\pmb{x},S)
=
\min_{\pmb{x}\in \Xset}\max_{S\in\Gamma} \sum_{t=1}^{T}
L_t(\mathbf{X}_t,\mathbf{D}_t(S)).
\]
In other words, we wish to find among all production plans the one that
minimizes the maximum production plan cost over all scenarios, 
that minimizes~$A(\pmb{x})$, $A(\pmb{x})$ is the \emph{maximal cost of production
plan}~$\pmb{x}$. 
An optimal solution $\pmb{x}^r$ to the problem~\textsc{ROB}
is called  \emph{optimal robust production plan}.
 
Let $\pmb{x}\in \Xset$ be a given production plan.
A scenario~$S^{o}\in \Gamma$ that minimizes  the total cost $F(\pmb{x},S)$ of 
the production plan $\pmb{x}$
is
called \emph{optimistic scenario}.
A scenario~$S^{w}\in \Gamma$ that maximizes the total cost 
 $F(\pmb{x},S)$ of the production plan
 $\pmb{x}$ is
called \emph{the
worst case scenario}.

\subsection{Evaluating  Production Plan}
\label{sepp}

In this section, we 
show how to evaluate a given production plan
$\pmb{x}^{*}\in \Xset$. 
We first consider the problem of determining 
the optimal interval, 
$F_{\pmb{x}^{*}}=[f_{\pmb{x}^{*}}^{-},f_{\pmb{x}^{*}}^{+}]$,
containing possible values of costs of the production plan~$\pmb{x}^{*}$
which can be rigorously defined as the following optimization problems:
\begin{align}
f_{\pmb{x}^{*}}^{-}=&\min_{S\in \Gamma}F(\pmb{x}^{*},S),\label{popts}\\
f_{\pmb{x}^{*}}^{+}=&\max_{S\in \Gamma}F(\pmb{x}^{*},S).\label{ppess}
\end{align}
It is easily seen that the problem of computing the optimal lower bound on costs 
of~$\pmb{x}^{*}$ (\ref{popts}) is equivalent to 
the one of determining   an optimistic 
scenario~$S^{o}$ for $\pmb{x}^{*}$, namely 
$f_{\pmb{x}^{*}}^{-}=F(\pmb{x}^{*},S^{o})= \min_{S\in \Gamma}F(\pmb{x}^{*},S)$.
Similarly, 
 the problem of computing the optimal upper bound on costs 
of~$\pmb{x}^{*}$ (\ref{ppess}) is equivalent to 
the problem of determining   a worst case
scenario~$S^{w}$ for $\pmb{x}^{*}$, i.e.
$f_{\pmb{x}^{*}}^{+}=
A(\pmb{x}^{*})=F(\pmb{x}^{*},S^{w})= \max_{S\in \Gamma}F(\pmb{x}^{*},S)$. Thus
\begin{equation}
F_{\pmb{x}^{*}}=[f_{\pmb{x}^{*}}^{-},f_{\pmb{x}^{*}}^{+}]=[F(\pmb{x}^{*},S^{o}), F(\pmb{x}^{*},S^{w})].
\label{icost}
\end{equation}

Using the optimal interval $F_{\pmb{x}^{*}}$ of possible values of costs of
production plan~$\pmb{x}^{*}$ 
allows us to evaluate  \emph{possibility} and \emph{necessity} that the cost of 
the plan does not exceed a given threshold under uncertain demands modeled by 
intervals. Hence,
in order to assert \emph{possibility}
 that the cost of 
the plan does not exceed a given threshold~$g$, i.e. to assert whether there exits a scenario $S\in \Gamma$ for which $F(\pmb{x}^{*},S)\leq g$,
it suffices to determine an optimistic scenario $S^{o}$, the optimal
lower bound $f_{\pmb{x}^{*}}^{-}=F(\pmb{x}^{*},S^{o})$ and
evaluate $f_{\pmb{x}^{*}}^{-}\leq g$. If the inequality holds then there exists a scenario;
otherwise not. Similarly,  evaluating
 \emph{necessity} that the cost of 
the plan does not exceed a given threshold~$g$,
i.e. asserting  whether  $F(\pmb{x}^{*},S)\leq g$
for every scenario $S\in \Gamma$,
we only need to determine worst 
case scenario~$S^{w}$,
the optimal
upper bound $f_{\pmb{x}^{*}}^{+}=F(\pmb{x}^{*},S^{w})$
and
evaluate $,f_{\pmb{x}^{*}}^{+}\leq g$. 
Thus, evaluating a production plan boils down to computing its optimistic and worst case
scenarios.

Let us consider the problem of computing an optimistic 
scenario for a given production plan $\pmb{x}^{*}\in \Xset$, that
is the problem (\ref{popts}).
Its minimum is attained for some $S\in\Gamma$,
since $F(\pmb{x},S)$ is a continuous function on the bounded closed set $\Gamma$.
Problem~(\ref{popts}) 
can be formulated as a linear programming problem:
\begin{equation}
 \begin{array}{llll}
 f_{\pmb{x}^{*}}^{-}=&\min &\sum_{t=1}^{T}(c^I_t I_t+ c^B_t B_t)&\\
&\text{s.t. } & B_t- I_t=\sum_{j=1}^{t}(s_j-x^{*}_j), & t=1,\ldots,T,\\
  &     &s_t\in [d^{-}_t,d^{+}_t],  &t=1,\ldots,T,\\
    &   &B_t,I_t\geq 0,  &t=1,\ldots,T.
 \end{array}
 \label{poptslp}
\end{equation}
If $s^{o}_t$, $B^{o}_t$ and $I^{o}_t$, $t=1,\ldots,T$,
 is an optimal solution to problem~(\ref{poptslp}), then
$S^{o}=(s^{o}_1,\ldots,s^{o}_T)$ is an optimistic scenario for $\pmb{x}^{*}$ 
(an optimistic realization of uncertain demands)
and 
$I^{o}_t$  is storing inventory amount
from period~$t$ to period~$t+1$ and 
$B^{o}_t$ represents backordering amount from period~$t+1$ to
period~$t$ under the optimistic scenario $S^{o}$.
The problem~(\ref{poptslp}) can be reduced to the classical minimum
cost flow problem and effectively solved by algorithms that take into 
a special structure of the underling network (see, e.g.,~\cite{AH08}).
Hence, and fact that $c^I_t,c^B_t\geq 0$ it follows that for $t=1,\ldots,T$
one of $I^{o}_t$ and $B^{o}_t$ is zero.

Let us study   the problem of computing a worst case 
scenario for a given production plan $\pmb{x}^{*}\in \Xset$, that
is the problem (\ref{ppess}).
Since $F(\pmb{x}^{*},S)$ is a continuous function on the bounded closed set~$\Gamma$,
it attains maximum  for some $S\in\Gamma$.
The problem~(\ref{ppess})
can be formulated as a mixed integer programming problem (MIP):
\begin{equation}
 \begin{array}{llll}
  f_{\pmb{x}^{*}}^{+}=&\max &\sum_{t=1}^{T}(c^I_t I_t+ c^B_t B_t)&\\
&\text{s.t. } & B_t- I_t=\sum_{j=1}^{t}(s_j-x^{*}_j), & t=1,\ldots,T,\\
  &     &s_t\in [d^{-}_t,d^{+}_t],  &t=1,\ldots,T,\\
    &    &B_t,I_t\geq 0,  &t=1,\ldots,T,\\
      &  &I_t\leq (1-\delta_t)\sum_{j=1}^{t}(x^{*}_j-d^{-}_j),  &t=1,\ldots,T,\\
      & &B_t\leq \delta_t\sum_{j=1}^{t}(d^{+}_j-x^{*}_j),  &t=1,\ldots,T,\\
      & & \delta_t\in \{0,1\}, &t=1,\ldots,T.      
 \end{array}
 \label{imostnec}
\end{equation}
Let $s^{w}_t$, $B^{w}_t$, $I^{w}_t$ and $\delta_t$, $t=1,\ldots,T$, be
 an optimal solution to problem~(\ref{imostnec}). Then
$S^{w}=(s^{w}_1,\ldots,s^{w}_T)$ is a worst case scenario for $\pmb{x}^{*}$ 
(a pessimistic realization of uncertain demands)
and 
$I^{w}_t$  is storing inventory amount
from period~$t$ to period~$t+1$ and 
$B^{w}_t$ is backordering amount from period~$t+1$ to
period~$t$ under the worst case scenario~$S^{w}$.
The last two constraints model~(\ref{imostnec})
and the binary variables $\delta_t$ ensure
that storing inventory  from period~$t$ to period~$t+1$
and backordering  from period~$t+1$ to
period~$t$ is not performed simultaneously 
(either $I_t>0$ or $B_t>0$). If $\delta_t=1$ then
backordering is performed $B_t>0$; otherwise 
storing inventory is performed $I_t>0$.
Thus, the problem~(\ref{ppess}) turns out to be much harder than (\ref{popts}).

We now solve the problem~(\ref{ppess})  by means of \emph{dynamic programming}.
Let us present a
result which shows that
determining a worst case scenario~$S^{w}$ 
can be restricted to the vertices of $\Gamma$, that is to the set
of  extreme scenarios~$\Gamma_{\text{ext}}$.
We prove the convexity of the cost function on $\Gamma$.
\begin{prop}
Function $F(\pmb{x}^{*},S)$ is convex on $\Gamma$ for any fixed production 
plan~$\pmb{x}^{*}\in \Xset$.
\label{ofcon}
\end{prop}
\begin{proof}
Function $c^I_t(\mathbf{X}^{*}_t-\mathbf{D}_t(S))$ and $c^{B}_t(\mathbf{D}_t(S)-\mathbf{X}^{*}_t)$
are convex on $\Gamma$. Then so are $\max\{c^{I}_t(\mathbf{X}^{*}_t-\mathbf{D}_t(S)),
c^{B}_t(\mathbf{D}_t(S)-\mathbf{X}^{*}_t)\}$
and $\sum_{t=1}^{T}\max\{c^{I}_t(\mathbf{X}^{*}_t-\mathbf{D}_t(S)),
c^{B}_t(\mathbf{D}_t(S)-\mathbf{X}^{*}_t)\}$.
\end{proof}
The following result allows us to 
reduce the set of scenarios~$\Gamma$  to the set
of  extreme scenarios~$\Gamma_{\text{ext}}$.
\begin{prop}
An optimal scenario for problem~(\ref{ppess}) (a worst case scenario) is an extreme one.
\label{pext}
\end{prop}
\begin{proof}
	Function $F(\pmb{x}^{*},S)$ attains its maximum in $\Gamma$.
	Since $F(\pmb{x}^{*},S)$ is convex (Proposition~\ref{ofcon}) 
	and $\Gamma$ is the hyper-rectangle,
	an optimal scenario 
	for problem~(\ref{ppess}) is an extreme one
	(see, e.g.,~\cite{M75}).
\end{proof}
Applying Proposition~\ref{pext}, we can rewrite problem~(\ref{ppess}) as:
\begin{equation}
f_{\pmb{x}^{*}}^{+}=
A(\pmb{x}^{*})=F(\pmb{x}^{*},S^{w})= \max_{S\in \Gamma_{\text{ext}}}F(\pmb{x}^{*},S).
\label{ppessex}
\end{equation}

We are now ready to give a \emph{dynamic programming based  algorithm} for solving 
problem~(\ref{ppessex}). Let  $\mathbb{D}_t$ be the set of feasible 
cumulative demand levels in period~$t$, $t=1,\ldots, T$, i.e.
$\mathbb{D}_t=\{\mathbf{D}_t(S^{-}), \mathbf{D}_t(S^{-})+1,\ldots, \mathbf{D}_t(S^{+})\}$,
let $\mathcal{L}_{t-1}(\mathbf{D})$ be the maximal cost of a given production plan~$\pmb{x}^{*}$
over periods $t,\ldots,T$, when the cumulative demand level up to period~$t-1$ is
equal to~$\mathbf{D}$, $\mathbf{D}\in \mathbb{D}_{t-1}$, 
$\mathcal{L}_{t-1}: \mathbb{D}_{t-1}\rightarrow\Rset_{+}$.
Set 
$\mathbb{D}_0=\{0\}$. It is evident that:
\begin{align}
\mathcal{L}_{T}(\mathbf{D})&=0  &\mathbf{D}\in \mathbb{D}_{T},\label{brec1}\\
\mathcal{L}_{t-1}(\mathbf{D})&=\max\left\{
\begin{array}{l}
L_t(\mathbf{X}^{*}_t,\mathbf{D}+d^{-}_t)+\mathcal{L}_{t}(\mathbf{D}+d^{-}_t)\label{brec2}\\
L_t(\mathbf{X}^{*}_t,\mathbf{D}+d^{+}_t)+\mathcal{L}_{t}(\mathbf{D}+d^{+}_t)
\end{array}  
\right\}
&\mathbf{D}\in \mathbb{D}_{t-1},\\
&&t=T,\ldots,1.\nonumber
\end{align}
The maximal cost of production plan~$\pmb{x}^{*}$ over period $1,\ldots,T$ is
equal to  $\mathcal{L}_{0}(0)$, $\mathcal{L}_{0}(0)=f_{\pmb{x}^{*}}^{+}$,
which is computed according to the backward recursion (\ref{brec1}) and (\ref{brec2}).
The corresponding to~$\pmb{x}^{*}$  worst case scenario~$S^{w}$ can be determined
by a forward recursion technique. It is sufficient to store for 
each $\mathbf{D}\in \mathbb{D}_{t-1}$ the value for which 
the maximum in (\ref{brec2}) is attained, that is
either $\mathbf{D}+d^{-}_t$ or $\mathbf{D}+d^{+}_t$.
The running time of the dynamic programming based  algorithm is
$O(T \cdot \mathbf{D}_T)$, which is pseudo-polynomial.
We have  thus proved the following theorem.
\begin{thm}
There is an algorithm 
for computing the maximal cost of 
a given production plan~$\pmb{x}^{*}$ and its a worst case scenario~$S^{w}$,
which runs in $O(T \cdot \mathbf{D}_T)$.
\end{thm}
Since finding a worst case scenario
requires  taking into account only
extreme demand scenarios (see Proposition~\ref{pext}),
the running time of the above algorithm may be additionally refined
by reducing the cardinalities  
of sets  $\mathbb{D}_{t}$, $t=1,\ldots,T$,
in (\ref{brec1}) and (\ref{brec2}).
Note that we need only consider 
cumulative demand levels~$\mathbf{D}$  which
can be obtained by summing  instantiated    demands, 
at  their lower or upper bounds, in the periods up to~$t$.
Namely,  $\mathbf{D}\in \mathbb{D}_{t}$
if and only if  $\mathbf{D}=\sum_{k=1}^{t} d_k$, where
$d_k\in \{d^{-}_k, d^{+}_k\}$.
Hence, each  reduced set
of possible  cumulative demand levels
$\mathbb{D}_{t}$ has form $\{\mathbf{D}^1_t,\ldots, \mathbf{D}^l_t\}$.
Of course, $l$ is bound  by
 $\mathbf{D}_t(S^{+})$.
 An idea of  the improved
dynamic programming based  algorithm can be outlined as follows:

\noindent \emph{Step~1 (reducing sets $\mathbb{D}_{t}$):} built
a directed acyclic network composed of node
that is  associated with  $\mathbb{D}_{0}=\{0\}$ and 
$T$ layers that are associated with 
the reduced sets 
of possible  cumulative demand levels
$\mathbb{D}_{t}=\{\mathbf{D}^1_t,\ldots, \mathbf{D}^l_t\}$,
 $t=1,\ldots,T$.
The $t$-th layer is composed of nodes
corresponding  to cumulative demand levels $\mathbf{D}$, 
$\mathbf{D}\in
\mathbb{D}_{t}$.
The arc between $\mathbf{D}^u_{t-1}$ and $\mathbf{D}^v_t$ exists if and only if  either
 $\mathbf{D}^v_t=\mathbf{D}^u_{t-1}+d^{-}_t$ or
$\mathbf{D}^v_t=\mathbf{D}^u_{t-1}+d^{+}_t$.
Observe that, 
except for nodes from the last
layer,
each node in the network has
exactly two outgoing arcs.
An example of the constructed network is presented in Figure~\ref{fig1}.

\noindent \emph{Step~2:} 
compute the maximal cost of production plan~$\pmb{x}^{*}$ over period $1,\ldots,T$ according to the backward recursion (\ref{brec1}) and (\ref{brec2}) in the constructed
network with the layers corresponding to
reduced sets $\mathbb{D}_{t}$
and store for each $\mathbf{D}\in \mathbb{D}_{t-1}$
 the value $\mathbf{D}\in \mathbb{D}_{t}$ for which 
the maximum in (\ref{brec2}) is attained.

\noindent \emph{Step~3:} determine a worst case scenario for
production plan~$\pmb{x}^{*}$ by performing a simple 
forward recursion in the constructed network using the stored (in Step~2)
values for which 
the maxima in (\ref{brec2}) are attained.
\begin{figure}
\centering
	\includegraphics{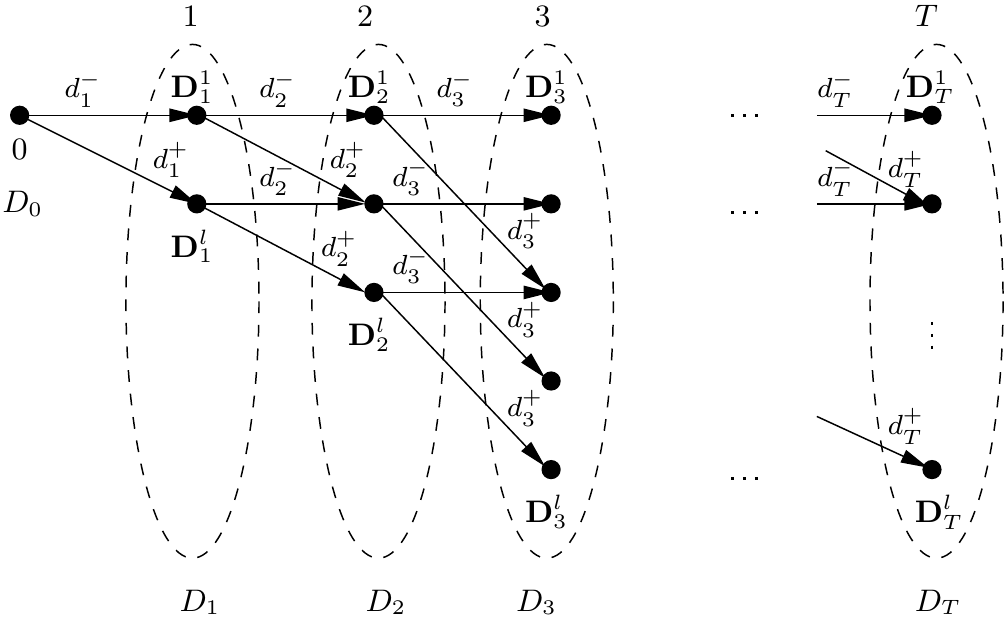}
	\caption{An example of the constructed network in Step~1.}\label{fig1}
\end{figure}

The network in Step~1 can be built in $O(T\cdot \max_{t=1,\ldots, T} |\mathbb{D}_{t}|)$ time.
The running time of Step~2 is  the same time as Step~1.
Step~3 can be done in $O(T)$ time. Hence, the overall running time of
the improved algorithm is $O(T\cdot \max_{t=1,\ldots, T} |\mathbb{D}_{t}|)$.
It is easily seen that $\max_{t=1,\ldots, T} |\mathbb{D}_{t}|$ is upper 
bounded by $\mathbf{D}_T(S^{+})$ and 
at the worst case $\max_{t=1,\ldots, T} |\mathbb{D}_{t}|=\mathbf{D}_T(S^{+})$.

Furthermore, the running time can be reduced if
$d^{+}_1-d^{-}_1=\cdots=d^{+}_T-d^{-}_T=h$.
Then we find 
$\mathbb{D}_t=\{\mathbf{D}_t(S^{-}), \mathbf{D}_t(S^{-})+h,\ldots,
 \mathbf{D}_t(S^{-})+t h\}$, $t=1,\ldots,T$.
 Now the running time is $O(T^{2})$, which is polynomial.

\subsection{Solving the Robust Problem}
\label{ssrp}

Let us consider the problem~\textsc{ROB} with no capacity
limits, i.e. the problem with the set 
$\Xset=\{(x_1,\ldots,x_{T})\;:\;
x_t\geq 0, t=1,\ldots,T\}$.
In this case, we make  the assumption: 
the costs of carrying one unit of 
inventory from period~$t$ to period~$t+1$ 
for every $t=1,\ldots,T$
are equal,
we denote it by $c^I$  and
the costs  of backordering one unit 
from period~$t+1$ to period~$t$ for every $t=1,\ldots,T$
are equal,
we denote it by $c^B$. 
 Note that
function
$F(\pmb{x},S)$ is continuous on $\Xset$ and $\Gamma$,
$\Gamma$ is a closed bounded set, and so $A(\pmb{x})$ is 
well  defined continuous function on $\Xset$ (see, e.g.,~\cite[Theorem~1.4]{M70}).
We show that 
an optimal robust production plan $\hat{\pmb{x}}=(\hat{x}_1,\ldots,\hat{x}_{T})$, i.e. $\hat{\pmb{x}}=
\text{arg }\min_{\pmb{x}\in\Xset}A(\pmb{x})$,
exists and 
can be
computed by the following formulae:
\begin{equation}
\renewcommand{\arraystretch}{1.6}
\begin{array}{lll}
\hat{\mathbf{X}}_1:=\frac{c^{B} \mathbf{D}_1(S^{+})+c^{I}\mathbf{D}_1(S^{-})}{c^{B}+c^{I}},
& \hat{x}_1:=\hat{\mathbf{X}}_1&\\
\hat{\mathbf{X}}_t:=\frac{c^{B} \mathbf{D}_t(S^{+})+c^I\mathbf{D}_t(S^{-})}{c^{B}+c^{I}},& \hat{x}_t:=\hat{\mathbf{X}}_t-\hat{\mathbf{X}}_{t-1},& t\geq 2.
\end{array}
\label{adopuc}
\end{equation}
An algorithm for determining a production plan~$\hat{\pmb{x}}$ according to (\ref{adopuc})
can be implemented in $O(T)$ time. 
Before we show that $\hat{\pmb{x}}$ is optimal to  problem~\textsc{ROB}
(an optimal robust production plan)
 with no capacity
limits we prove the following proposition.
\begin{prop}
Let $\hat{\pmb{x}}$ be a plan computed according to (\ref{adopuc}). Then
$\hat{\pmb{x}}$ is feasible and
$S^{-}$ and $S^{+}$ are the worst case scenarios for 
plan~$\hat{\pmb{x}}$, i.e.
$A(\hat{\pmb{x}})=F(\hat{\pmb{x}},S^{-})=F(\hat{\pmb{x}},S^{+})$.
\label{pwcs}
\end{prop}
\begin{proof}
See Appendix~\ref{dod}.
\end{proof}
We are now ready to prove 
that $\hat{\pmb{x}}$ is an optimal robust production plan.
\begin{thm}
A production plan determined by formulae~(\ref{adopuc}) is
an optimal one for problem~\textsc{ROB} with no capacity
limits.
\label{tncrob}
\end{thm}
\begin{proof}
See Appendix~\ref{dod}.
\end{proof}
Note that 
if   an initial backorder~$B_0$ or an initial inventory~$I_0$
are not equal to zero then one can modify the interval demand $D_1$ as follows:
$D_1:=[d^{-}_1+B_0, d^{+}_1+B_0]$  if $B_0>0$ or
$D_1:=[d^{-}_1-I_0, d^{+}_1-I_0]$ if $0<I_0\leq d^{-}_1$,
and apply 
formulae~(\ref{adopuc}) to determine an optimal robust production plan for  problem~\textsc{ROB}
with the modified demand.
If $I_0> d^{-}_1$ then one appends period~$0$, as it has been described in Section~\ref{sdlsp},
and  applies an algorithm (Algorithm~\ref{aorp}) for the case with  capacity
limits, $l_t=0$, $u_t=M$, $t=1,\ldots,T$, where $M$ is a large number.

Let us turn to  the problem~\textsc{ROB} with  capacity
limits, i.e. the problem with the set 
$\Xset=\{(x_1,\ldots,x_{T})\;:\;
l_t\leq x_t\leq u_t, t=1,\ldots,T\}$.
Notice $\Gamma$ is a bounded closed set.
Function $F(\pmb{x},S)$ is continuous on 
$\Xset$ and $\Gamma$ and hence $A(\pmb{x})$ is continuous
function on $\Xset$ (see, e.g.,~\cite[Theorem~1.4]{M70}). From this and  the fact
$\Xset$ is a bounded closed set it follows that 
$A(\pmb{x})$ attains its minimum on $\Xset$.

We now construct an iterative algorithm for solving problem~\textsc{ROB} based on
on iterative relaxation  scheme for min-max problems proposed in~\cite{SA80}.
Similar methods were developed for min-max regret linear programming problems
with an interval objective function~\cite{IS95,ML99}.
Let us consider the problem (\textsc{RX-ROB})
being a relaxation of  problem~\textsc{ROB}  that consists in 
replacing a given scenario set~$\Gamma$ with a discrete scenario set
$\Gamma_{\text{dis}}=\{S^1,\ldots,S^K\}$, $\Gamma_{\text{dis}}\subseteq \Gamma$:
\begin{equation}
 \begin{array}{rll}
  \textsc{RX-ROB:}&a^{*}=\min a&\\
  \text{s.t. }&a\geq F(\pmb{x},S^{k})&\forall S^k\in \Gamma_{\text{dis}},\\
  &\pmb{x}\in \Xset,&
 \end{array}
 \label{relrob}
\end{equation}
where $S^{k}=(s^{k}_t)_{t=1}^{T}$.
The constraint 
$a\geq F(\pmb{x},S^{k})$,
called \emph{scenario cut}, is associated with exactly one scenario~$S^k\in \Gamma_{\text{dis}}$.
Since $\Gamma_{\text{dis}}\subseteq \Gamma$, the maximal cost~$a^{*}$ of
an optimal solution $\pmb{x}^{*}$ of problem~\textsc{RX-ROB} 
over discrete scenario set~$\Gamma_{\text{dis}}$
is a lower bound on
the maximal cost of an optimal robust production plan~$\pmb{x}^{r}$ for
problem~\textsc{ROB}, i.e. $a^{*}\leq A(\pmb{x}^{r})$.
Note that the scenario cut, $a\geq F(\pmb{x},S^{k})$,
associated with~$S^{k}$
 is not a linear constraint. One can linearize the cut
by replacing it in \textsc{RX-ROB} with the following $T+1$ constraints 
and $2T$ new decision variables:
\[
 \begin{array}{ll}
a\geq \sum_{t=1}^{T}(c^I_t I^{S^{k}}_t+ c^B_t B^{S^{k}}_t),&\\
B^{S^{k}}_t- I^{S^{k}}_t=\sum_{j=1}^{t}(s^{k}_j-x_j),&t=1,\ldots,T, \\
         B^{S^{k}}_t,I^{S^{k}}_t\geq 0,&t=1,\ldots,T.
 \end{array}
\]

Our algorithm (Algorithm~\ref{aorp}) starts with zero 
lower bound on the maximal cost of an optimal robust production plan~$\pmb{x}^{r}$, $LB=0$,
a candidate $\pmb{x}^{*}\in\Xset$ for an optimal solution for \textsc{ROB} and
empty discrete scenario set, $\Gamma_{\text{dis}}=\emptyset$. 
At each iteration, a worst case scenario $S^{w}$ 
for $\pmb{x}^{*}$ is computed by applying the method~(\ref{imostnec}) 
or the dynamic programming based algorithm presented in Section~\ref{sepp}.
Clearly, $A(\pmb{x}^{*})=F(\pmb{x}^{*},S^{w})$ is an upper bound on $A(\pmb{x}^{r})$,
$A(\pmb{x}^{r})\leq A(\pmb{x}^{*})$. If a termination criterion is fulfilled (usually
$(A(\pmb{x}^{*})-LB)/LB\leq \epsilon$ if $LB>1$; $A(\pmb{x}^{*})-LB\leq \epsilon$
otherwise, $\epsilon>0$ is a given tolerance) then algorithm stops with 
production plan~$\pmb{x}^{*}$, which is an approximation of 
an optimal robust production plan.
Otherwise
the worst case scenario $S^{w}=(s^{w}_t)_{t=1}^{T}$ is added to $\Gamma_{\text{dis}}$,
the corresponding to~$S^{w}$ scenario cut  
is appended to
problem~\textsc{RX-ROB}. Next
the updated linear programming problem~\textsc{RX-ROB} 
 is solved to obtain a better 
candidate~$\pmb{x}^{*}$ for an optimal solution for \textsc{ROB} and
new lower bound $LB=a^{*}$. Since set $\Gamma_{\text{dis}}$ is 
updated during the course of the algorithm,
the computed values of lower bounds are nondecreasing sequence of their values.
 Then new iteration is started. 
  \begin{algorithm}
  \KwIn{ Interval demands $D_{t}=[d^{-}_t, d^{+}_t]$, costs~$c^I_t$, $c^B_t$,
  $t=1,\ldots T$,
   initial production plan~$\pmb{x}^{*}\in \Xset$,
  a convergence tolerance parameter~$\epsilon>0$.}
  \KwOut{A production plan~$\hat{\pmb{x}}^{r}$, an approximation of 
  an optimal robust production plan, and its worst case scenario~$S^{w}$.}

\KwStep~0. $k:=0$, $LB:=0$, $\Gamma_{\text{dis}}:=\emptyset$.\\
\KwStep~1. $\pmb{x}^{k}:=\pmb{x}^{*}$.\\ 
\KwStep~2. Compute a worst case scenario $S^{w}$ 
                    for $\pmb{x}^{k}$ by applying the method~(\ref{imostnec}) 
                    or the dynamic programming based algorithm presented in Section~\ref{sepp}.\\                   
\KwStep~3. $\Delta:=F(\pmb{x}^{k}, S^{w})- LB$. If $LB>1$ then $\Delta:=\Delta/LB$. 

 If $\Delta \leq \epsilon$ then output $\pmb{x}^{k}, S^{w}$ and STOP.\\	
\KwStep~4.  $k:=k+1$.\\		
\KwStep~5. $S^{k}:=S^{w}$, $\Gamma_{\text{dis}}:=\Gamma_{\text{dis}}\cup \{S^{k}\}$ and
     append scenario   cut $a\geq F(\pmb{x},S^{k})$
      to
     problem~\textsc{RX-ROB}.\\
\KwStep~6. Compute an optimal solution~$(\pmb{x}^{*},a^{*})$ for~\textsc{RX-ROB},
                   $LB:=a^{*}$, and go to Step~1.               
  \caption{Solving problem~\textsc{ROB}.}
  \label{aorp}
\end{algorithm}

In order to choose a good initial production plan~$\pmb{x}^{*}\in \Xset$ in Algorithm~\ref{aorp},
we suggest to solve the classical production planning problem (\ref{ddls}) with capacity limits 
(the model~(\ref{ddlsf})) under the midpoint demand scenario~$S^{m}$,
i.e. $d_t(S^{m})=(d^{-}_t+d^{+}_t)/2$, $t=1,\ldots,T$ and take an optimal production plan under the
midpoint scenario as an  initial production plan.
\begin{thm}
Algorithm~\ref{aorp} terminates in a finite number of steps for any given~$\epsilon>0$.
\label{tconv}
\end{thm}
\begin{proof}
See Appendix~\ref{dod}.
\end{proof}
Note that 
if an initial inventory~$I_0$ or an initial backorder~$B_0$
are not equal to zero then one appends period~$0$, as it has been described in Section~\ref{sdlsp},
and  applies Algorithm~\ref{aorp} for $T+1$ periods.

Let us illustrate, by the following example, 
that solving  problem~\textsc{ROB}  leads to a robust  production
plan.   We are given $5$ periods with 
the production capacity limits on~a production plan:
$l_1=40$, $u_1=50$,  $l_2=30$, $u_2=40$,
$l_3=30$, $u_3=40$, $l_4=10$, $u_4=35$ and $l_5=10$, $u_5=35$.
The costs of carrying one unit
of inventory from period~$t$ to period~$t+1$, $c^I_t$, for every
$t=1,\ldots,5$  equal~$1$ and  
 the costs of
backordering one unit from period $t+1$ to period~$t$, $c^B_t$, for every
 $t=1,\ldots,5$ equal~$5$.
The knowledge about
 demands in each period is represented by the intervals: $D_1=[30,45]$, 
 $D_2=[5,15]$,  $D_3=[10,30]$,  $D_4=[20,40]$ and  $D_5=[20,40]$.
 The scenario set~$\Gamma$ (states of the world) is 
$\Gamma= [30,45]\times [5,15]\times [10,30] \times [20,40] \times [20,40]$
(see Figure~\ref{fig2}).
The execution of Algorithm~\ref{aorp} gives a  production plan:
$x^{\text{opt}}_1=40$, $x^{\text{opt}}_2=30$,  $x^{\text{opt}}_3=30$,
$x^{\text{opt}}_4=27.9167$,  $x^{\text{opt}}_5=10$ with the total cost of~$215.833$
($\pmb{x}^{\text{opt}}$ is an approximation of  an optimal robust production plan with 
convergence tolerance parameter~$\epsilon=0.0001$, 
the maximal cost of~$\pmb{x}^{\text{opt}}$ is no more than $0.01\%$
from optimality).
The worst case scenario $S^{w}\in \Gamma$ is $d_1(S^{w})=30$, 
$d_2(S^{w})=5$, $d_3(S^{w})=10$, $d_4(S^{w})=20$, $d_5(S^{w})=20$
 and
$\max_{S\in \Gamma}F(\pmb{x}^{\text{opt}},S)=F(\pmb{x}^{\text{opt}},S^{w})=215.833$
(see Figure~\ref{fig2}).
This means that  total costs of production plan~$\pmb{x}^{\text{opt}}$ 
do not exceed the value of $215.833$ over the set of scenarios.
Moreover,  the plan~$\pmb{x}^{\text{opt}}$ has the best worst performance, i.e.
it minimizes the total cost over the all scenarios. 
Additionally, making use of the methods presented in Section~\ref{sepp},
one gets  complete information about all
 possible values of costs of the production plan~$\pmb{x}^{\text{opt}}$
 over the set of scenarios~$\Gamma$,
 by determining 
the optimal interval
 $F_{\pmb{x}^{\text{opt}}}=[f_{\pmb{x}^{\text{opt}}}^{-},f_{\pmb{x}^{\text{opt}}}^{+}]$  
 that contains
 these values (see (\ref{icost})). This interval equals $[40, 215.833]$.
A~popular approach for solving a problem with uncertain  parameters
 modeled by the classical intervals is taking the midpoints of the intervals (the average values of
 the possible parameter values)
 and solving the problem with these deterministic parameters.
 In our example, the midpoint scenario has the form 
 $d_1(S^{\text{mid}})=37.5$, 
$d_2(S^{\text{mid}})=10$, $d_3(S^{\text{mid}})=20$, $d_4(S^{\text{mid}})=30$,
 $d_5(S^{\text{mid}})=30$. An algorithm for the problem with 
 the midpoint demands (see Section~\ref{sdlsp}) returns an optimal production
 plan:
 $x^{\text{mid}}_1=40$, $x^{\text{mid}}_2=30$,  $x^{\text{opt}}_3=30$,
$x^{\text{mid}}_4=10$,  $x^{\text{mid}}_5=17.5$ with the total cost of~$70$. But,
if  scenario 
$d_1(S^{w})=45$, 
$d_2(S^{w})=15$, $d_3(S^{w})=30$, $d_4(S^{w})=40$,
 $d_5(S^{w})=40$ (a worst case scenario for $\pmb{x}^{\text{mid}}$) occurs 
 then the cost will be equal to~$357.5$ ($\max_{S\in \Gamma}F(\pmb{x}^{\text{mid}},S)=357.5$).
 Note that $357.5\gg 215.833$.
Similar situation is for two extreme scenarios $S^{+}$ and $S^{-}$,  the scenarios
  in which all the demands are set to their upper bounds and  the lower bounds, respectively.
  Again running an algorithm for the crisp problem with the demands under scenario~$S^{+}$
  and $S^{-}$, we obtain optimal solutions:
   $x^{+}_1=45$, $x^{+}_2=30$,  $x^{+}_3=30$,
$x^{+}_4=30$,  $x^{+}_5=35$ with the total cost of~$35$ under $S^{+}$
and
   $x^{-}_1=40$, $x^{-}_2=30$,  $x^{-}_3=30$,
$x^{-}_4=10$,  $x^{-}_5=10$, under $S^{-}$,
with the total cost of~$180$. It turns out that
if  scenario 
$d_1(S^{w})=30$, 
$d_2(S^{w})=5$, $d_3(S^{w})=10$, $d_4(S^{w})=20$,
 $d_5(S^{w})=20$, $S^w\in\Gamma$, (a worst case scenario for $\pmb{x}^{+}$) occurs
 then  the cost of~ $\pmb{x}^{+}$ will be equal to~$270$. Similarly,
 if 
$d_1(S^{w})=45$, 
$d_2(S^{w})=15$, $d_3(S^{w})=30$, $d_4(S^{w})=40$,
 $d_5(S^{w})=40$, $S^w\in\Gamma$,  (a worst case scenario for $\pmb{x}^{-}$) occurs
 then  the cost of~ $\pmb{x}^{-}$ will be equal to~$395$. 
 The  summary of the input and output data  of the above illustrative example
 is shown in Table~\ref{tabsum1}.
 Accordingly, we have no doubts that
the computed plan~$\pmb{x}^{\text{opt}}$ with respect to the
min-max criterion  (problem~\textsc{ROB}) is a robust one. 
\begin{figure}
\centering
            \includegraphics{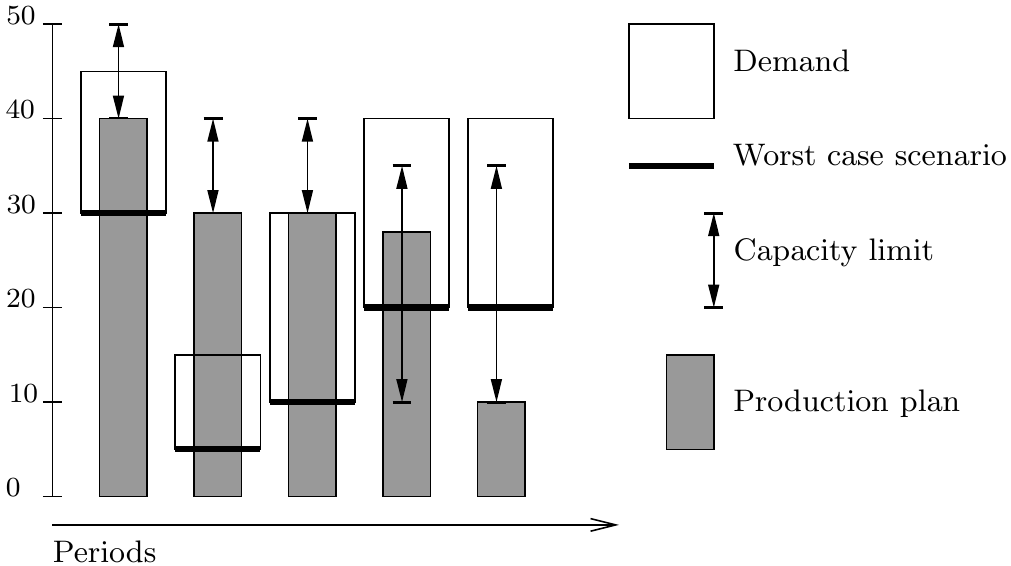}
	\caption{The computed  production plan~$\pmb{x}^{\text{opt}}$ in the
	illustrative example.}\label{fig2}
\end{figure}

\begin{table}
\begin{small}
\setlength{\tabcolsep}{3pt}
\caption{Summary of the data  and results of the illustrative example} \label{tabsum1}
\begin{center}
  \begin{tabular}{cccc ||    cc  |    cc  | cc  |  cc}
            &\multicolumn{3}{c||}{data$^{\dag}$}&\multicolumn{8}{c}{results}\\
    \cline{2-12}
           &\multicolumn{2}{c}{capacity}&interval&\multicolumn{2}{c|}{robust}&\multicolumn{2}{c|}{plan} 
           &\multicolumn{2}{c|}{plan}&\multicolumn{2}{c}{plan}\\
           &\multicolumn{2}{c}{limits}&demands&\multicolumn{2}{c|}{plan}
           &\multicolumn{2}{c|}{under $S^{\text{mid}}$} &\multicolumn{2}{c|}{under $S^{+}$}&
           \multicolumn{2}{c}{under $S^{-}$}\\
     \hline      
    $t$ &$l_t$  & $u_t$ & $D_t$ & $x^{\text{opt}}_t$& $d_t(S^{w})$ & $x^{\text{mid}}_t$ & $d_t(S^{w})$ 
    & $x^{+}_t$ &  $d_t(S^{w})$ & $x^{-}_t$ & $d_t(S^{w})$ \\ 
    \hline
    1 & 40 & 50 & [30,45] & 40 & 30         & 40   & 45 & 45 & 30 & 40 & 45 \\ 
    2 & 30 & 40 & [5,15] & 30 & 5             & 30   & 15 & 30 & 5   & 30 & 15 \\ 
    3 & 30 & 40 & [10,30] & 30& 10          & 30   & 30 & 30 & 10 & 30 & 30 \\ 
    4 & 10 & 35 & [20,40] & 27.9167 & 20& 10   & 40 & 30 & 20 & 10 & 40  \\ 
    5 & 10 & 35 & [20,40] & 10 & 20         & 17.5& 40 & 35 & 20 & 10 & 40 \\ 
    \hline
    \multicolumn{4}{l ||}{\footnotesize{$^{\dag}c^I_t=1, c^B_t=5, t=1,\ldots,5$}}&
    \multicolumn{8}{c}{the worst costs for plans: $\pmb{x}^{\text{opt}}$, 
    $\pmb{x}^{\text{mid}}$, $\pmb{x}^{+}$, $\pmb{x}^{-}$ }\\
     \cline{5-12}
     \multicolumn{4}{c||}{}&\multicolumn{2}{c|}{$F(\pmb{x}^{\text{opt}},S^w)$}
     &\multicolumn{2}{c|}{$F(\pmb{x}^{\text{mid}},S^w)$}&
     \multicolumn{2}{c|}{$F(\pmb{x}^{+},S^w)$}&\multicolumn{2}{c}{$F(\pmb{x}^{-},S^w)$}\\
      \cline{5-12}
     \multicolumn{4}{c||}{}&\multicolumn{2}{c|}{215.833}
     &\multicolumn{2}{c|}{357.5}&
     \multicolumn{2}{c|}{270}&\multicolumn{2}{c}{395}\\  
  \end{tabular}
\end{center}
\end{small}
\end{table}

In order to check the efficiency of Algorithm~\ref{aorp},  we  performed some computational tests.
For every $T=100, 200, \ldots, 1000$, ten instances of the problem~\textsc{ROB} 
 with capacity limits
were generated.
 In every instance, inventory costs were randomly chosen from the set $\{1, 2, \ldots, 10\}$, backorder costs were randomly chosen from the set $\{20, 21, \ldots, 50\}$, the demands  and the production capacities
 were randomly generated  intervals $[X, Y]$, where $X$ is an integer-valued  random variable uniformly
 distributed in $\{0, 1, \ldots, 99\}$ and $Y$ is an integer-valued random variable uniformly
  distributed in $\{100, 101, \ldots, 199\}$. 
 To solve   the generated instances, we used  IBM ILOG CPLEX 12.2 library 
 (parallel using up 2 threads)~\cite{CPLEX} 
 and a
 computer equipped with Intel Core 2 Duo 2.5 GHz.
In Table~\ref{tab1new}, minimal, average and maximal computation times in seconds,
required to find approximations of  optimal robust production plans with 
convergence tolerance parameter~$\epsilon=0.0001$,
are presented. Thus,
 the maximal costs of the computed production plans are no more than $0.01\%$
from optimality.
All computations finished in a few iterations and about  $98\%$ of the total running time was spent on computing  worst case scenarios by MIP model (\ref{imostnec}).
We also  implemented and ran the improved dynamic programming based algorithm 
for computing  worst case scenarios presented in Section~\ref{sepp}, but
it turned out that solving MIP model for determining worst case scenarios
was much faster than  computing them
 by the dynamic programming algorithm.
As we can see from the obtained results, Algorithm~\ref{aorp} allows us to solve quite large problems having up to 1000 periods in reasonable time.

\begin{table}
\begin{small}
  \setlength{\tabcolsep}{4pt}
\caption{Minimal, average and maximal computation times in seconds} 
\label{tab1new}
\begin{center}
\begin{tabular}{r|rrrrrrrrrr}
T & 100 &  200 & 300 & 400 & 500 &
 600 &  700 & 800 & 900 & 1000\\
\hline 
min & 0.25 & 0.63 & 0.96 & 2.04 & 2.64 
               & 23.55 & 9.64 & 24.54 & 61.26 &  41.44  \\ 
avg & 0.39 & 0.84 & 5.68 & 9.18 & 11.47 
                & 61.44 & 72.20 & 180.78  & 282.20 &  372.58 \\ 
max &0.65 & 1.05 & 13.40 & 20.47 & 21.78 
                & 147.32 & 188.36 & 328.61 & 573.26 & 893.01 
\end{tabular}
\end{center}
\end{small}
\end{table}

\section{Fuzzy Problem}
\label{sfp}

In this section, we apply a more elaborate approach to model 
uncertain demands. 
Namely, 
the  uncertain demands, in problem~(\ref{ddls}), are
modeled by fuzzy intervals $\widetilde{D}_t$, $t=1,\ldots,T$.
Here,
a membership function of~$\widetilde{D}_t$ is regarded as
a possibility distribution
for the values of the unknown demand~$d_t$ (see Section~\ref{ssnpt}).
The possibility degree of the assignment 
$d_t=s$ is
$\Pi(d_t=s)=
\pi_{d_t}(s)=\mu_{\widetilde{D}_t}(s)$.
Let $S=(s_t)_{t=1}^T$
be a scenario that represents a  state of the world
where
$d_t=s_t$, for $t=1,\ldots,T$.
It is assumed that 
the demands are unrelated  one to each other. 
Hence,
the 
possibility distributions
 associated with the demands induce the following possibility distribution over all scenarios in $S\in \Rset^T$ (see~\cite{DFG03}):
\begin{equation}
\pi(S)=\Pi((d_1=s_1)\wedge\cdots\wedge (d_T=s_T))
 =\min_{t=1,\ldots,T}\Pi(d_t=s_t)
 =\min_{t=1,\ldots,T}\mu_{\widetilde{D}_t}(s_t).\label{pscen}
 \end{equation}
 The value of $\pi(S)$ stands for the possibility of the event that 
 scenario~$S\in \Rset^T$ will occur.
 We have thus extended scenario set~$\Gamma$ given by the intervals  
 (see Section~\ref{srob}) to the fuzzy case and now
 $\widetilde{\Gamma}$ is a fuzzy
 set of scenarios with membership function
 $\mu_{\widetilde{\Gamma}}(S)=\pi(S)$, $S\in\Rset^T$.
We see at once that the  $\lambda$-cuts of $\widetilde{\Gamma}$  for every $\lambda\in (0,1]$
fulfill the following equality:
\[
\widetilde{\Gamma}^{[\lambda]}=
\{S:\; \pi(S)\geq \lambda\}=[d^{-[\lambda]}_1,d^{+[\lambda]}_1]\times\cdots\times
 [d^{-[\lambda]}_T,d^{+[\lambda]}_T],
\]
which is from~(\ref{pscen}) and the definition of $\lambda$-cut.
We also define $\widetilde{\Gamma}^{[0]}
=[d^{-[0]}_1,d^{+[0]}_1]\times\cdots\times
 [d^{-[0]}_T,d^{+[0]}_T]$.
Notice that $\widetilde{\Gamma}^\lambda$, $\lambda \in [0,1]$, is the classical scenario set containing all scenarios whose possibility of occurrence is not less than~$\lambda$.

\subsection{Evaluating  Production Plan}

In order to  choose a reasonable production plan under fuzziness, we 
first show how to evaluate a given production plan~$\pmb{x}\in \Xset$.
Notice that a cost of  production plan~$\pmb{x}$ is unknown quantity, denoted by $f_{\pmb{x}}$, since demands are unknown and
modeled by fuzzy intervals in the setting of possibility theory.
Thus, the unknown cost~$f_{\pmb{x}}$ falls within fuzzy interval~$\widetilde{F}_{\pmb{x}}$, called
\emph{fuzzy cost} of plan~$\pmb{x}$, whose 
 membership function~$\mu_{\widetilde{F}_{\pmb{x}}}$
is a  possibility distribution for the 
  values of  fuzzy variable~$f_{\pmb{x}}$
  (unknown cost of~$\pmb{x}$), $\pi_{\widetilde{F}_{\pmb{x}}}=\mu_{\widetilde{F}_{\pmb{x}}}$,
defined as follows:
\begin{equation}
\mu_{\widetilde{F}_{\pmb{x}}}(v)=\Pi(f_{\pmb{x}}=v)=
\sup_{\{S: \;F(\pmb{x},S)=v\}}\pi(S), \;\;v\in\Rset. \label{dfc}
\end{equation}
Making use of (\ref{dfc}), we can define \emph{degrees of possibility and
necessity} that a cost of 
a given plan  $\pmb{x}\in \Xset$  does not exceed a given threshold~$g$:
\begin{align}
\Pi(f_{\pmb{x}}\leq g)&=\sup_{v\leq g}\pi_{\widetilde{F}_{\pmb{x}}}(v)=
\sup_{\{S: \;F(\pmb{x},S)\leq g\}}\pi(S),\label{pdx}\\
\mathrm{N}(f_{\pmb{x}}\leq g)&= 1-\Pi(f_{\pmb{x}}> g)=
1-\sup_{v>g}\pi_{\widetilde{F}_{\pmb{x}}}(v)
=1-\sup_{\{S: \;F(\pmb{x},S)> g\}}\pi(S).\label{ndx}
\end{align}
It is easily seen that $\Pi(f_{\pmb{x}}\leq g)=\lambda$ means that there exists a scenario~$S$
such that $\pi(S)=\lambda$ in which the cost of plan $\pmb{x}$ 
does not exceed threshold~$g$, $F(\pmb{x},S)\leq g$.
$\mathrm{N}(f_{\pmb{x}}\leq g)=1-\lambda$
means that for all scenarios~$S$
such that $\pi(S)>\lambda$, costs of plan $\pmb{x}$ under these scenarios 
do not exceed~$g$.

We now consider the problem of computing
the degrees~(\ref{pdx}) and (\ref{ndx})
 of a given plan~$\pmb{x}$. 
 Write $\widetilde{F}_{\pmb{x}}^{[\lambda]}=
[f^{-[\lambda]}_{\pmb{x}},f^{+[\lambda]}_{\pmb{x}}]$. Note that the interval $\widetilde{F}_{\pmb{x}}^{[\lambda]}$
($\lambda$-cut of the fuzzy cost) is the optimal 
interval of possible costs of $\pmb{x}$ (see~(\ref{icost}))
in problem~\textsc{ROB} under interval scenario set $\widetilde{\Gamma}^{[\lambda]}$.
Hence there exists a link between the interval and the fuzzy cases:
\begin{eqnarray}
\Pi(f_{\pmb{x}}\leq g)&=&
\sup\{\lambda\in [0,1]: f^{-[\lambda]}_{\pmb{x}}\leq g\},\label{pdxi}\\
\mathrm{N}(f_{\pmb{x}}\leq g)&=&1-\inf\{\lambda
	\in[0,1]:   f^{+[\lambda]}_{\pmb{x}}\leq g\}.\label{ndxi}
\end{eqnarray}
From equations  (\ref{pdxi}) and (\ref{ndxi}), we obtain methods for
computing the degrees. 
So, in order 
to compute 
$\Pi(f_{\pmb{x}}\leq g)$ (resp. $\mathrm{N}(f_{\pmb{x}}\leq g)$)
we need to find the largest  (resp. smallest) value of $\lambda$ such that
there exits a scenario $S\in \widetilde{\Gamma}^{[\lambda]}$ for which $F(\pmb{x},S)\leq g$
(resp. 
for every scenario $S\in \widetilde{\Gamma}^{[\lambda]}$ inequality $F(\pmb{x},S)\leq g$ holds),
which is equivalent to
determine an optimistic scenario $S^{o}\in \widetilde{\Gamma}^{[\lambda]}$ 
(resp. a worst 
case scenario~$S^{w}\in \widetilde{\Gamma}^{[\lambda]}$)
for $\pmb{x}$
by solving~(\ref{poptslp}) (resp.~(\ref{imostnec})) and
evaluating $F(\pmb{x},S^{o})\leq g$
(resp. $F(\pmb{x},S^{w})\leq g$). Notice $f^{-[\lambda]}_{\pmb{x}}=F(\pmb{x},S^{o})$
(resp. $f^{+[\lambda]}_{\pmb{x}}=F(\pmb{x},S^{w})$).
Since $f^{-[\lambda]}_{\pmb{x}}$  (resp. $f^{+[\lambda]}_{\pmb{x}}$) is nondecreasing 
(resp. nonincreasing)
function
of $\lambda$,
we can apply a binary search technique on $\lambda \in [0,1]$.

The fuzzy cost of production plan~$\pmb{x}$ 
(the possibility distribution for costs of~$\pmb{x}$), $\widetilde{F}_{\pmb{x}}$
can be determined approximately, if necessary,
via the use of $\lambda$-cuts. Namely, 
the optimal intervals of possible costs of $\widetilde{F}^{[\lambda]}_{\pmb{x}}=
[f^{-[\lambda]}_{\pmb{x}}, f^{+[\lambda]}_{\pmb{x}}]$ under $\widetilde{\Gamma}^{[\lambda]}$
are computed for
suitably chosen $\lambda$-cuts.
Then  fuzzy cost~$\widetilde{F}_{\pmb{x}}$ is reconstructed from their
$\lambda$-cuts. This
approach makes sense since intervals $\widetilde{F}^{[\lambda]}_{\pmb{x}}$ are nested.

\subsection{Fuzzy Robust Problem}

We now propose  two  criteria of choosing a
robust solution in the fuzzy-valued
problem~(\ref{ddls}).

We are given a threshold~$g$, and we would like to find a production plan
which maximizes 
the degree of certainty (necessity)  that its cost does not exceed  threshold~$g$. Thus, we
would like to solve the following problem:
\begin{equation}
	\label{nectr}
	\max_{\pmb{x}\in \Xset}\mathrm{N}(f_{\pmb{x}}\leq g).
\end{equation}
There are no doubts that an optimal production plan computed according to~(\ref{nectr}) is
a robust one, since with the highest
degree of certainty  costs of the plan over scenarios 
will not exceed threshold~$g$. 
By (\ref{ndxi}), it is easy to check that problem~(\ref{nectr}) is equivalent to the following 
mathematical programming problem:
\begin{equation}
 \begin{array}{ll}
 \min &\lambda\\
 \text{s.t.} & f^{+[\lambda]}_{\pmb{x}}\leq 
   g,\\
       &\lambda\in [0,1],\\
       & \pmb{x}\in \Xset.
 \end{array}
 \label{nectrm}
\end{equation}
If $\lambda^{*}$ is the optimal objective and $\pmb{x}^{*}$ is an optimal solution for
problem~(\ref{nectrm}) 
then $\mathrm{N}(f_{\pmb{x}^{*}}\leq g)=1-\lambda^{*}$
and $\pmb{x}^{*}$ is an optimal production plan for~(\ref{nectr}).
 If~(\ref{nectrm}) is
infeasible then 
$\mathrm{N}(f_{\pmb{x}}\leq g)=0$ for all $\pmb{x}\in \Xset$.

We now present a more general criterion of choosing a robust production plan than~(\ref{nectr}). Namely,
suppose that a decision maker knows her/his preferences about  
a cost of a production plan
$f_{\pmb{x}}$ and expresses it by
a \emph{fuzzy goal}~$\widetilde{G}$, which is a fuzzy interval with a bounded support
and a nonincreasing  upper 
semicontinuous membership function~$\mu_{\widetilde{G}}:\Rset\rightarrow [0,1]$
such that $\mu_{\widetilde{G}}(v)=1$
for $v\in [0,g]$.
The value of $\mu_{\widetilde{G}}(f_{\pmb{x}})$ is  the extent to which
cost $f_{\pmb{x}}$ of $\pmb{x}$ satisfies the decision maker.
Now the requirement ``$f_{\pmb{x}}\leq g$"  is replaced softer one, i.e. ``$f_{\pmb{x}}\in \widetilde{G}$''. 
So, by~(\ref{dnec}) and (\ref{dfc})
the necessity that event ``$f_{\pmb{x}}\in \widetilde{G}$'' holds can be expressed  as follows:
\begin{eqnarray}
\mathrm{N}(f_{\pmb{x}}\in \widetilde{G})
&=&1-\Pi(f_{\pmb{x}}\not\in \widetilde{G})\label{ndg}\\
&=&
1-\sup_{v\in \Rset}\min\{\pi_{f_{\pmb{x}}}(v),1-\mu_{\widetilde{G}}(v)\}\nonumber\\
&=&1-\sup_{S}\min\{\pi(S),1-\mu_{\widetilde{G}}(F(\pmb{x},S))\}.\nonumber
\end{eqnarray}
Thus, if $\mathrm{N}(f_{\pmb{x}}\in \widetilde{G})=1-\lambda$
means that for all scenarios~$S$
such that $\pi(S)>\lambda$, 
the degree that costs of plan~$\pmb{x}$ fall within 
fuzzy goal~$\widetilde{G}$,
is not less than $1-\lambda$. Note that $\mathrm{N}(f_{\pmb{x}}\in \widetilde{G})$
is more general  and weaker than $\mathrm{N}(f_{\pmb{x}}\leq g)$. If 
 $\mu_{\widetilde{G}}(v)=0$ for $v>g$ then they are the same. Moreover,
  $\mathrm{N}(f_{\pmb{x}}\leq g)\leq \mathrm{N}(f_{\pmb{x}}\in \widetilde{G})$.
 
Let us give the second criterion of choosing a robust plan.
We are given a fuzzy goal $\widetilde{G}$, and we wish to find a production plan
which maximizes the necessity degree that costs of the plan fall within 
fuzzy goal~$\widetilde{G}$. Thus we need to solve the following optimization problem,
\begin{equation}
	\label{nectrg}
	\max_{\pmb{x}\in \Xset}\mathrm{N}(f_{\pmb{x}}\in \widetilde{G}).
\end{equation}
We check at once that problem~(\ref{nectrg}) 
 is equivalent to the following 
mathematical programming problem, which is from (\ref{ndxi}) and (\ref{ndg}):
\begin{equation}
 \begin{array}{ll}
 \min &\lambda\\
 \text{s.t.} & f^{+[\lambda]}_{\pmb{x}}\leq 
   g^{+[1-\lambda]},\\
       &\lambda\in [0,1],\\
       & \pmb{x}\in \Xset.
 \end{array}
 \label{nectrmg}
\end{equation}
If $(\pmb{x}^{*},\lambda^{*})$ is an optimal solution
for problem~(\ref{nectrmg}), then $\mathrm{N}(f_{\pmb{x}^{*}}\in \widetilde{G})=1-\lambda^{*}$.
 If~(\ref{nectrmg}) is
infeasible then 
$\mathrm{N}(f_{\pmb{x}}\in \widetilde{G})=0$ for all $\pmb{x}\in \Xset$. 

An algorithm for solving problem~(\ref{nectrmg}) (resp. (\ref{nectrm}))
is based on the standard  binary search 
technique in $[0,1]$ (the interval of possible values of $\lambda$)
 which follows from the fact that 
$f^{+[\lambda]}_{\pmb{x}}$ is nonincreasing and
 $g^{+[1-\lambda]}$ (resp.~$g$) is nondecreasing function of $\lambda$. We call
 the algorithm the \emph{binary search based algorithm}.
To find an optimal  $(x^{*},\lambda^{*})$, $x^{*}\in\Xset$, $\lambda^{*}\in [0,1]$, with a given 
error tolerance~$\xi>0$, we seek
at each iteration, for a fixed $\lambda$,
a plan $\pmb{x}\in \Xset$ satisfying 
$f^{+[\lambda]}_{\pmb{x}}\leq g^{+[1-\lambda]}$ (resp.
$f^{+[\lambda]}_{\pmb{x}}\leq g$), which boils down to 
seeking an optimal robust production plan $\pmb{x}^{r}$ and its worst case scenario
under scenario set~$\widetilde{\Gamma}^{[\lambda]}$,
i.e. to solving problem~\textsc{ROB} (see Section~\ref{srob}). 
Note that $f^{+[\lambda]}_{\pmb{x}}=A(\pmb{x})$
and also that 
$f^{+[\lambda]}_{\pmb{x}}\leq g^{+[1-\lambda]}$ (resp.
$f^{+[\lambda]}_{\pmb{x}}\leq g$) holds for some $\pmb{x}\in \Xset$
 if and only if it holds for an optimal robust production plan under $\widetilde{\Gamma}^{[\lambda]}$.
Thus, at each iteration, we use  either formulae~(\ref{adopuc}) for the case without capacity limits or
Algorithm~\ref{aorp} for the case with capacity limits.
If the length of determined interval of possible values of~$\lambda$ is less than $\xi$, then
an optimal robust production plan $\pmb{x}^{r}$ for a fixed $\lambda$, $(\pmb{x}^{r},\lambda)$,  is
an approximation of
an optimal solution for
problem~(\ref{nectrmg}) (resp. (\ref{nectrm})) with precision~$\xi$.
The running time of the above algorithm is
$O(I(T)\log\xi^{-1})$ time, where $\xi>0$ is a given
error tolerance and $I(|T|)$ is  time required 
for finding 
an optimal robust production plan  and its worst case scenario
under interval scenario set~$\widetilde{\Gamma}^{[\lambda]}$ 
(the running time of either~(\ref{adopuc}) or Algorithm~\ref{aorp}).

We now show, by the following illustrative example, that determining a production
plan maximizing  the necessity degree that costs of the plan fall within 
fuzzy goal~$\widetilde{G}$ (problem~(\ref{nectrg})) is a robust one under uncertain
demands modeled by fuzzy intervals.
We are given $5$ periods with 
the production capacity limits on~a production plan:
$l_1=40$, $u_1=50$,  $l_2=30$, $u_2=40$,
$l_3=30$, $u_3=40$, $l_4=10$, $u_4=35$ and $l_5=10$, $u_5=35$
and the same 
the costs of carrying one unit
of inventory from period~$t$ to period~$t+1$, $c^I_t=1$,
$t=1,\ldots,5$, and the same 
the costs of
backordering one unit from period $t+1$ to period~$t$ ,$c^B_t=5$, $t=1,\ldots,5$.
The demand uncertainty in each period is represented by the 
triangular fuzzy intervals: $\widetilde{D}_1=(30,37.5,45)$, 
 $\widetilde{D}_2=(5,10,15)$,  $\widetilde{D}_3=(10,20,30)$, 
  $\widetilde{D}_4=(20,30,40)$ and  $\widetilde{D}_5=(20,30,40)$,
  regarded as
possibility distributions
for the values of the unknown demands.
The fuzzy set of scenarios has the membership function:
$\mu_{\widetilde{\Gamma}}(S)=\pi(S)=\min_{t=1,\ldots,5}\mu_{\widetilde{D}_t}(s_t)$,
 $S\in\Rset^5$. The fuzzy goal $\widetilde{G}$ is   trapezoidal  fuzzy interval 
$\widetilde{G}=(0,0, 195.83,215.42)$,
where  $195.833$ is the maximal cost of an optimal robust production 
plan for the problem~\text{ROB} without capacity limits 
(an ideal supplier)
and under  the supports (the interval demands)
of the fuzzy demands, i.e. $\widetilde{D}^{[0]}_t$, $t=1,\ldots 5$.
Thus, a production plan with the cost less than~$195.833$ is totally accepted and
with the cost greater than~$215.42$ is not at all accepted.
The binary search based algorithm outputs 
a production plan (with ~$\xi=0.01$ and~$\epsilon=0.0001$ for Algorithm~\ref{aorp}):
$x^{\text{opt}}_1=40$, $x^{\text{opt}}_2=30$,  $x^{\text{opt}}_3=30$,
$x^{\text{opt}}_4=25.3776$,  $x^{\text{opt}}_5=10$,
that maximizes the necessity degree that costs of the plan fall within 
fuzzy goal~$\widetilde{G}$, $\mathrm{N}(f_{\pmb{x}^{\text{opt}}}\in \widetilde{G})=1-\lambda=0.883$.
This means that
 for all scenarios~$S$ whose possibility of occurrence is greater than~$0.117$,
 $\pi(S)>\lambda=0.117$, 
the  degree of necessity (the degree of certainty) that total  costs of plan~$\pmb{x}^{\text{opt}}$ fall within 
fuzzy goal~$\widetilde{G}$, is not less than~$0.883$. Furthermore, we are sure that
the total costs of the plan
do not exceed~$196.986$ (the total cost  at $\lambda$-cut equal to~$0.117$)
for every scenario~$S$ such that~$\pi(S)>0.117$.
We now apply existing approaches to our example.
We first consider methods based on a defuzzification which take into account 
only one scenario resulting  from a defuzzification of fuzzy parameters (demands)
-- see, e.g.,~\cite{PMPV09,PMJB10}.
Applying, for instance, the index of Yager~\cite{Y81},
we get crisp demands: 
 $d_1(S^{\text{Y}})=37.5$, 
$d_2(S^{\text{Y}})=10$, $d_3(S^{\text{Y}})=20$, $d_4(S^{\text{Y}})=30$,
 $d_5(S^{\text{Y}})=30$. An algorithm for the crisp dynamic lot-size problem with
 these demands (see Section~\ref{sdlsp}) returns an optimal production
 plan:
 $x^{\text{Y}}_1=40$, $x^{\text{Y}}_2=30$,  $x^{\text{Y}}_3=30$,
$x^{\text{Y}}_4=10$,  $x^{\text{Y}}_5=17.5$ with the total cost of~$70$.
However,  the cost of $\pmb{x}^{\text{Y}}$  may be even $313.262$ for scenarios~$S$
such that~$\pi(S)>0.117$, and so $313.262\gg 196.986$.
The necessity degree that costs of $\pmb{x}^{\text{Y}}$ fall within 
$\widetilde{G}$ equals 0.593, $\mathrm{N}(f_{\pmb{x}^{\text{Y}}}\in \widetilde{G})=0.593$,
which gives
$\mathrm{N}(f_{\pmb{x}^{\text{Y}}}\in \widetilde{G})<\mathrm{N}(f_{\pmb{x}^{\text{opt}}}\in \widetilde{G})$.
Let us examine 
a possibilistic programming  (a~mathematical programming with fuzzy parameters),
where  solution concepts are based on the Bellman-Zadeh approach~\cite{BZ70}
- see, e.g.,~\cite{MPP10,TRGZ07}. In this way,
the assertion of the form ``$f_{\pmb{x}}\in \widetilde{G}$'', 
where $f_{\pmb{x}}$ is a cost of production plan~$\pmb{x}$, 
is treated as a fuzzy constraint and  values of the membership function~$\mu_{\widetilde{G}}$
stand for
degrees of satisfaction of this constraint - the fuzzy goal. In other words, the assertion induces 
$\mathrm{\Pi}(f_{\pmb{x}}\in \widetilde{G})$.
The joint possibility distribution generated  by the fuzzy goal as well as constraints with fuzzy demands
is the minimum of the possibility of satisfaction of the fuzzy goal and the possibility of
feasibility of the constraints and thus an optimal production plan is a plan, denoted by $\pmb{x}^{\text{BZ}}$,
that maximizes the possibility degree of satisfaction both  goal and the constraints.
A trivial verification shows that for    production
plan~$\pmb{x}^{\text{Y}}$  with  the total cost of~$70$,
the possibility degree of satisfaction the goal as well as the constraints
is equal to~1  in our fuzzy problem under consideration -
in problem~(\ref{ddls}) with triangular fuzzy demands~$\widetilde{D}_t$,
 $t=1,\ldots,5$,  and so $\pmb{x}^{\text{BZ}}=\pmb{x}^{\text{Y}}$.
 However, as we have seen above, this plan is not a robust one.
 The summary of the data and results of the example is given in Table~\ref{tabsum2}.

 \begin{table}
\begin{small}
\setlength{\tabcolsep}{4pt}
\caption{Summary of the data  and results of the illustrative example} \label{tabsum2}
\begin{center}
  \begin{tabular}{cccc ||    c  |    c  | c  }
            &\multicolumn{3}{c||}{data$^{\dag}$}&\multicolumn{3}{c}{results}\\
    \cline{2-7}
           &\multicolumn{2}{c}{capacity}&interval&\multicolumn{1}{c|}{robust}&\multicolumn{1}{c|}{plan} 
           &\multicolumn{1}{c}{plan}\\
           &\multicolumn{2}{c}{limits}&demands&\multicolumn{1}{c|}{plan}
           &\multicolumn{1}{c|}{(index of Yager)} &\multicolumn{1}{c}{(Bellman-Zadeh)}\\
     \hline      
    $t$ &$l_t$  & $u_t$ & $ \widetilde{D}_t$ 
    & $x^{\text{opt}}_t$& $x^{\text{Y}}_t$  
    & $x^{\text{BZ}}_t$  \\ 
    \hline
    1 & 40 & 50 &(30,37.5,45)&  40&40 & 40\\ 
    2 & 30 & 40 & (5,10,15) &30 & 30&30\\ 
    3 & 30 & 40 & (10,20,30)&  30 & 30&30\\ 
    4 & 10 & 35 & (20,30,40) & 25.3776 & 10& 10\\ 
    5 & 10 & 35 & (20,30,40) &  10 & 17.5&17.5\\ 
    \hline
    \multicolumn{4}{l ||}{\footnotesize{$^{\dag}c^I_t=1, c^B_t=5, t=1,\ldots,5$}}&
    \multicolumn{3}{c}{the worst costs under $S$ such that $\pi(S)>0.117$ for}\\
    \cline{5-7}
    \multicolumn{4}{c||}{$\widetilde{G}=(0,0, 195.83,215.42)$}&\multicolumn{1}{c|}{$\pmb{x}^{\text{opt}}$}
     &\multicolumn{1}{c|}{$\pmb{x}^{\text{Y}}$}&
     \multicolumn{1}{c}{$\pmb{x}^{\text{BZ}}$}\\
      \cline{5-7}
     \multicolumn{4}{c||}{}&\multicolumn{1}{c|}{196.986}
     &\multicolumn{1}{c|}{313.262}&
     \multicolumn{1}{c}{313.262}\\  
      \cline{5-7}
         \multicolumn{4}{l ||}{}&
    \multicolumn{3}{c}{}\\
    \cline{5-7}
    \multicolumn{4}{c||}{}&\multicolumn{1}{c|}{$\mathrm{N}(f_{\pmb{x}^{\text{opt}}}\in \widetilde{G})$}
     &\multicolumn{1}{c|}{$\mathrm{N}(f_{\pmb{x}^{\text{Y}}}\in \widetilde{G})$}&
     \multicolumn{1}{c}{$\mathrm{N}(f_{\pmb{x}^{\text{BZ}}}\in \widetilde{G})$}\\
      \cline{5-7}
     \multicolumn{4}{c||}{}&\multicolumn{1}{c|}{0.883}
     &\multicolumn{1}{c|}{0.593}&
     \multicolumn{1}{c}{0.593}\\  
  \end{tabular}
\end{center}
\end{small}
\end{table}

In order to evaluate the efficiency of the binary search based algorithm,  we show some results of computational
experiments.
For every number of periods $T=100, 200, \ldots, 1000$, ten instances of the problem~(\ref{nectrg}) with capacity limits
were generated. In every instance, inventory costs were randomly chosen from the set $\{1, 2, \ldots, 10\}$, backorder costs were randomly chosen from the set $\{20, 21, \ldots, 50\}$, the production capacities were randomly generated intervals $[X, Y]$, where $X$ is an 
integer-valued  random variable uniformly distributed in $\{0, 1, \ldots, 99\}$ and $Y$ is an integer-valued random variable uniformly distributed in $\{100, 101, \ldots, 199\}$,
the demands are  triangular fuzzy intervals with the supports equal to $[0, 199]$ and the modal values equal 
to $Z$, where $Z$ is an integer-valued random variable uniformly distributed in $\{0, 1, \ldots, 199\}$,
the fuzzy goal $\widetilde{G}$ was modeled as a trapezoidal  fuzzy interval 
$\widetilde{G}=(0,0, c,d)$, where $c$ 
was chosen as
the maximal cost of an optimal robust production 
plan for the problem~\text{ROB} without capacity limits 
and under the interval demands being the supports 
of the generated triangular fuzzy demands.
The values of $d$ were equal to $c+\beta$, for $\beta\in
\{0.00 \cdot c, 0.25 \cdot c, 0.50 \cdot c, 0.75 \cdot c, 1.00 \cdot c\}$
and the error tolerance~$\xi=0.01$.
We used  IBM ILOG CPLEX 12.2 library  (parallel using up 2 threads)~\cite{CPLEX}
 and a
 computer equipped with Intel Core 2 Duo 2.5 GHz to solve the generated instances.
In Table~\ref{tab2new} average  computation times in seconds are presented.
As we can see from the obtained results, 
the binary search based algorithm, which calls Algorithm~\ref{aorp}
at each iteration with 
convergence tolerance parameter~$\epsilon=0.0001$, 
can solve efficiently the problem~(\ref{nectrg}),  with capacity limits, having up to~1000 periods.
\begin{table}
\begin{small}
\setlength{\tabcolsep}{4pt}
\caption{Average computation times in seconds} \label{tab2new}
\begin{center}
\begin{tabular}{r|rrrrr}
&\multicolumn{5}{c}{$\beta$}\\
 \cline{2-6}
T & $0.00 \cdot c$&  $0.25\cdot c$ & $0.50\cdot c$ & $0.75\cdot c$& $1.00\cdot c$ \\
\hline 
100 & 0.96 &  0.92 &   0.88 &   0.83 &  0.92 \\ 
200 &  2.36  &   2.29 &  2.26  &  2.11  & 2.10   \\ 
300 &  7.22 &  6.93 &  6.81 &  6.76 &  6.86 \\ 
400 &  14.68 &  14.67 &  14.55 &  14.57 & 14.55  \\ 
500 &  27.51 &  27.06 &  26.60 & 26.69  & 26.64  \\ 
600 &  46.73 &  46.85 &  46.72 &  46.38 &  46.49 \\ \
700 &  111.72 &  109.70 & 109.41  & 109.71  & 109.43  \\ 
800 &  223.22 & 217.24  & 217.89  &  217.65 & 217.85 \\ 
900 &  243.51 &  243.28 &  243.87 &  244.75 &  245.21 \\ 
1000 &  417.78 &  416.86 &  416.14 &  414.14 & 415.80  \\
\end{tabular}
\end{center}
\end{small}
\end{table}

\section{Conclusion}
In this paper, we have proposed  methods to compute a robust procurement plan in the collaborative supply chain, where the customer uses a version of MRP with  ill-known  demands
to  plan a production.
This problem is   a certain  version of 
the lot sizing problem
with  ill-known  demands modeled by
fuzzy intervals, whose membership functions 
are   regarded as possibility distributions for the values of the unknown demands.
We have introduced,  in this setting,
the degrees of possibility and
necessity that   the cost of a plan does   not exceed a given threshold and
a degree of necessity
that costs of a plan fall within 
a given fuzzy goal, which allows us to evaluate a given production plan. Moreover, we have provided
methods for computing these degrees.
For finding robust production plans under fuzzy demands,
we have proposed two criteria:
the first one consists in choosing
 a production plan
which maximizes 
the degree of necessity  that its cost does not exceed  a given threshold,
the second criterion is softer than the first one and consists  in choosing a plan with
 the maximum degree of necessity
that costs of the plan fall within 
a given fuzzy goal. 
We have constructed the algorithms for determining optimal robust production plans 
 with respect to the criteria and confirmed their efficiency experimentally.
The criteria are
a generalization, to the fuzzy case, of the
known from literature the min-max criterion.
Consequently, we have shown in the paper that there exists a link between  
interval uncertainty with the min-max criterion and  possibilistic uncertainty 
with the necessity based criteria.
It turns out that 
the  evaluation of a production plan  and choosing a plan in the fuzzy-valued 
problem are not harder than in the interval-valued case. 
The  difficulty of solving the fuzzy problems lies in the interval case, since it is reduced to
 solving a small number of interval problems.
Therefore,
we have discussed first the interval-valued case.
In this case, we 
have considered the problem of determining the optimal interval of possible costs of
a production plan, which allowed us to evaluate  the plan.
Determining the optimal bounds  of the interval
 boils down to  computing optimistic and worst case scenarios.
We have proposed  linear programing based method for computing an optimistic 
scenario and mixed integer programing and dynamic programming methods for
computing a worst case scenario. We have also identified a polynomial solvable case.
For computing an optimal robust production plan, we have provided 
a polynomial algorithm and iterative one for the cases: with no capacity limits and with
capacity limits, respectively.
Then
 we have extended  the methods
 introduced for the interval-valued problem
to the fuzzy-valued one. 

There is still  an open question concerning 
the complexity status of computing a worst case scenario
of a given production plan.
The problem
 is pseudopolynomially  solvable and polynomially solvable under certain assumptions and
seems to be
a core  of  most of the problems considered in the paper.
These assumptions are nearly realistic and  make possible 
extension of our approach to 
the case where a
procurement plan is given for a family of product.
In other words, when the sum of quantities
procured has to respect  supplier capacity constraints  which are computed from a previous
procurement plan. This problem  is equivalent to the multi-item capacitated lot sizing problem.
The fact that the complexity status is still open
 creates the possibility
 to find a polynomial algorithm
and to  extend our approach to the multi-item, multi-level capacitated lot sizing problem without
the assumptions.
So, it is 
an interesting topic of further
research.

\appendix

\section{Appendix}
\label{dod}

\begin{proof}[Proposition~\ref{pwcs}] 
It is easily seen that 
$\hat{\mathbf{X}}_t \in  [\mathbf{D}_t(S^{-}),\mathbf{D}_t(S^{+})]$,
$c^{I}(\hat{\mathbf{X}}_t-\mathbf{D}_t(S^{-}))=
c^{B}(\mathbf{D}_t(S^{+})-\hat{\mathbf{X}}_t)$, $t=1,\ldots,T$, and $\hat{x}_1\geq 0$.
Since $\mathbf{D}_{t-1}(S^{-})\leq \mathbf{D}_{t}(S^{-})$ and 
$\mathbf{D}_{t-1}(S^{+})\leq \mathbf{D}_{t}(S^{+})$, $t\geq 2$,
 (\ref{adopuc})
shows that $\hat{x}_t\geq 0$, $t\geq 2$.
Hence
$L_t(\hat{\mathbf{X}}_t,\mathbf{D}_t(S^{-}))=c^{I}(\hat{\mathbf{X}}_t-\mathbf{D}_t(S^{-}))$
and 
$L_t(\hat{\mathbf{X}}_t,\mathbf{D}_t(S^{+}))=c^{B}(\mathbf{D}_t(S^{+})-\hat{\mathbf{X}}_t)$, 
$t=1,\ldots,T$, and so $F(\hat{\pmb{x}},S^{-})=F(\hat{\pmb{x}},S^{+})$.
Let $S$ be any scenario, $S\in\Gamma$.  
Therefore,
$F(\hat{\pmb{x}},S)=\sum_{t=1}^{T}\max\{c^{I}(\hat{\mathbf{X}}_t-\mathbf{D}_t(S)),
c^{B}(\mathbf{D}_t(S)-\hat{\mathbf{X}}_t)\}
\leq\sum_{t=1}^{T}\max\{c^{I}(\hat{\mathbf{X}}_t-\mathbf{D}_t(S^{-})),
c^{B}(\mathbf{D}_t(S^{+})-\hat{\mathbf{X}}_t)\}
=\sum_{t=1}^{T}L_t(\hat{\mathbf{X}}_t,\mathbf{D}_t(S^{-}))=
\sum_{t=1}^{T}L_t(\hat{\mathbf{X}}_t,\mathbf{D}_t(S^{+}))\\
=F(\hat{\pmb{x}},S^{-})=F(\hat{\pmb{x}},S^{+})$.
\end{proof}

\begin{proof}[Theorem~\ref{tncrob}]
We show that $A(\pmb{x})\geq A(\hat{\pmb{x}})$ for every $\pmb{x}\in \Xset$.
Consider any $\pmb{x}^{'}\in \Xset$. Let us modify $\pmb{x}^{'}$ in the following way:
\[
\mathbf{X}^{''}_t=
\begin{cases}
	\mathbf{D}_t(S^{-}) &\text{if $\mathbf{X}^{'}_t<\mathbf{D}_t(S^{-})$,}\\
	\mathbf{X}^{'}_t &\text{if $\mathbf{D}_t(S^{-})\leq \mathbf{X}^{'}_t \leq \mathbf{D}_t(S^{+})$, $t=1,\ldots, T$,}\\
	\mathbf{D}_t(S^{+}) &\text{if $\mathbf{X}^{'}_t>\mathbf{D}_t(S^{+})$.}
\end{cases}
\]
Now $\mathbf{X}^{''}_t \in  [\mathbf{D}_t(S^{-}),\mathbf{D}_t(S^{+})]$.
From the feasibility of $\pmb{x}^{'}$, it follows that 
 $\mathbf{X}^{'}_{t-1} \leq  \mathbf{X}^{'}_t$, $t\geq 2$.
 Hence and  
$\mathbf{D}_{t-1}(S^{-})\leq \mathbf{D}_{t}(S^{-})$ and 
$\mathbf{D}_{t-1}(S^{+})\leq \mathbf{D}_{t}(S^{+})$, $t\geq 2$,
we obtain $\mathbf{X}^{''}_{t-1} \leq  \mathbf{X}^{''}_t$, $t\geq 2$ and, in
consequence $\pmb{x}^{''}\in \Xset$. Furthermore, it is easy to check that
$A(\pmb{x}^{'})\geq A(\pmb{x}^{''})$. By the definition of $A(\pmb{x}^{''})$, we have
$A(\pmb{x}^{''})\geq\max\{F(\pmb{x}^{''},S^{-}), F(\pmb{x}^{''},S^{+})\}$.
We now only need to show that $\max\{F(\pmb{x}^{''},S^{-}), F(\pmb{x}^{''},S^{+})\}\geq 
F(\hat{\pmb{x}},S^{-})=F(\hat{\pmb{x}},S^{+})$.

Let us focus on a production plan~$\hat{\pmb{x}}$ (determined by formulae~(\ref{adopuc})).
Notice that function $F(\pmb{x},S^{-})=\sum_{t=1}^{T}c^{I}(\mathbf{X}_t-\mathbf{D}_t(S^{-}))$
(resp. $F(\pmb{x},S^{+})=\sum_{t=1}^{T}c^{B}(\mathbf{D}_t(S^{+})-\mathbf{X}_t)$)
are the sum of  linear and increasing  (resp. decreasing) functions 
 with respect to $\mathbf{X}_t$,  $c^{I}(\mathbf{X}_t-\mathbf{D}_t(S^{-}))\geq 0$
 (resp. $c^{B}(\mathbf{D}_t(S^{+})-\mathbf{X}_t)\geq 0$) for
 $\mathbf{X}_t \in  [\mathbf{D}_t(S^{-}),\mathbf{D}_t(S^{+})]$, $t=1,\ldots,T$.
 Hence, for each $t=1,\ldots,T$ there exists the intersection point in interval
 $[\mathbf{D}_t(S^{-}),\mathbf{D}_t(S^{+})]$.
 It is easy to check that 
$\hat{\mathbf{X}}_t \in  [\mathbf{D}_t(S^{-}),\mathbf{D}_t(S^{+})]$ and
$c^{I}(\hat{\mathbf{X}}_t-\mathbf{D}_t(S^{-}))=
c^{B}(\mathbf{D}_t(S^{+})-\hat{\mathbf{X}}_t)$, $t=1,\ldots,T$.
Therefore, the points $\hat{\mathbf{X}}_t$,  $t=1,\ldots,T$,
 are the intersection ones and
$\max\{F(\pmb{x}^{''},S^{-}), F(\pmb{x}^{''},S^{+})\}\geq 
F(\hat{\pmb{x}},S^{-})=F(\hat{\pmb{x}},S^{+})$. Proposition~\ref{pwcs}
shows that $F(\hat{\pmb{x}},S^{-})=F(\hat{\pmb{x}},S^{+})=A(\hat{\pmb{x}})$ and we thus get
$A(\pmb{x}^{'})\geq A(\pmb{x}^{''})\geq A(\hat{\pmb{x}})$.
 \end{proof}

\begin{proof}[Theorem~\ref{tconv}]
The proof is almost the same as those given in~\cite[Theorem~2.5]{G72} and \cite[Theorem~3]{SA80}.
We will denote by $\{(\pmb{x}^{k},a^{k})\}$ the sequence of optimal solution~$(\pmb{x}^{*},a^{*})$
computed in consecutive  iterations in Step~6 of Algorithm~\ref{aorp},
 $k$ stands for the  $k$-th iteration.
 By picking a subsequence, if necessary, the sequence $\{(\pmb{x}^{k},a^{k})\}$
converges to point $(\hat{\pmb{x}},\hat{a})$, $\hat{\pmb{x}}\in \Xset$, which 
follows from the fact that sequence~$\{\pmb{x}^{k}\}$ belong to the bounded and closet
 set~$\Xset \subseteq \Rset^{T}_{+}$ ($\Xset$ is a compact set) and
$\{a^{k}\}$ is a nondecreasing sequence bounded above. Similar considerations apply to the 
sequence $\{S^k\}$ of worst case scenarios determined in Step~2 of Algorithm~\ref{aorp}.
 The set $\Gamma$ is
a closed and bounded (compact) and hence  $\{S^k\}$  converges to $\hat{S}\in \Gamma$.
Since scenario cuts are appended to  problem~\textsc{RX-ROB}, the inequality 
$a^{k+1}\geq F(\pmb{x}^{k+1},S^{k})$ holds. By continuity of $F$,
we have
\begin{equation}
\hat{a}\geq F(\hat{\pmb{x}},\hat{S}).
\label{in1}
\end{equation}
Let us define the set $\mathcal{S}^{w}(\pmb{x})$ of worst case scenarios for $\pmb{x}\in X$, i.e.
$\mathcal{S}^{w}(\pmb{x})=\{S^{w}\,|\,S^{w}=\text{arg }\max_{S\in \Gamma}F(\pmb{x},S)\}$, which is
the point-to-set mapping. The set $\mathcal{S}^{w}(\pmb{x})$ is nonempty for every for $\pmb{x}\in X$.
By \cite[Theorem~1.5]{M70}, $\mathcal{S}^{w}$ is upper semicontinuous at $\hat{\pmb{x}}$
and so $\hat{S}\in \mathcal{S}^{w}(\hat{\pmb{x}})$. Therefore 
\begin{equation}
   A(\hat{\pmb{x}})=\max_{S\in \Gamma} F(\hat{\pmb{x}},S)=F(\hat{\pmb{x}},\hat{S}).
   \label{in2}
\end{equation}
Combining~(\ref{in1}) and (\ref{in2}), we obtain $\hat{a}\geq A(\hat{\pmb{x}})$.
By \cite[Lemma~1.2]{M70} $A$ is upper semicontinuous, which  yields
\[
A(\pmb{x}^{k})= F(\pmb{x}^{k},S^{k}) \leq a^{k}+\epsilon, \text{ for some sufficiently large } k.
\]
This implies that the termination criterion in Step~3 of Algorithm~\ref{aorp} will be  fulfilled in a finite number of iterations.
\end{proof}


\end{document}